\documentclass[sigconf]{acmart}
\usepackage[short]{optional} 
\usepackage{caption}
\usepackage[labelformat=simple]{subcaption}

\usepackage{graphicx,xcolor} 
\usepackage{bbm}
\usepackage{mathtools}
\usepackage{amsmath,amsthm}
\usepackage{mathrsfs}
\usepackage{url}
\usepackage{svg}
\usepackage{braket}
\usepackage{hyperref}
\usepackage[most]{tcolorbox}
\usepackage[ruled,vlined]{algorithm2e}
\usepackage[framemethod=tikz]{mdframed}
\usepackage{graphicx,pifont}
\usepackage{xcolor,colortbl}
\usepackage{multirow}
\usepackage{pifont}
\usepackage{xcolor}
\usepackage{makecell}
\usepackage{colortbl}
\usepackage{booktabs}
\usepackage{arydshln}
\usepackage{rotating}
\usepackage{balance}

\theoremstyle{plain}
\newtheorem{thm}{Theorem}[section]

\newcommand\norm[1]{\left\lVert#1\right\rVert}
\usepackage{framed}
\definecolor{shadecolor}{rgb}{0.88,0.93,0.93}
\definecolor{darkred}{rgb}{0.5,0.3,0.7}

\newif\ifshowcomments
\showcommentstrue 

\ifshowcomments
\newcommand{\mynote}[2]{\fbox{\bfseries\sffamily\scriptsize{#1}}
	{\small$\blacktriangleright$\textsf{\emph{#2}}$\blacktriangleleft$}}
\else
\newcommand{\mynote}[2]{}
\fi

\definecolor{lgray}{gray}{0.95}
\newcommand{\commentout}[1]{}
\definecolor{cadmiumgreen}{rgb}{0.0, 0.42, 0.24}

\newcommand{\seeding}{Disparity Seeding\xspace }

\let\oldding\ding
\renewcommand{\ding}[2][1]{\scalebox{#1}{\oldding{#2}}}

\newcommand{\xmark}{\ding[1]{55}}

\copyrightyear{2021}
\acmYear{2021}
\setcopyright{acmlicensed}
\acmConference[CIKM '21] {Proceedings of the 30th ACM International Conference on Information and Knowledge Management}{November 1--5, 2021}{Virtual Event, Australia.}
\acmBooktitle{Proceedings of the 30th ACM Int'l Conf. on Information and Knowledge Management (CIKM '21), November 1--5, 2021, Virtual Event, Australia}
\acmPrice{15.00}
\acmISBN{978-1-4503-8446-9/21/11}
\acmDOI{10.1145/3459637.3482375}

\begin{document}
\fancyhead{} 
\setlength{\textfloatsep}{5pt}
\setlength{\floatsep}{0pt}

\title{On Influencing the Influential: \seeding}

\author{Ya-Wen Teng$^{1,2}$, Hsi-Wen Chen$^{1}$, De-Nian Yang$^{1,2}$}
\author{Yvonne-Anne Pignolet$^{3}$, Ting-Wei Li$^{1}$,  Lydia Chen$^{4}$} 
\affiliation{$^1$\textit{Institute of Information Science, Academia Sinica, Taiwan}\country{}}
\affiliation{$^2$\textit{Research Center for Information Technology Innovation, Academia Sinica, Taiwan}\country{}}
\affiliation{$^3$\textit{DFINITY, Switzerland}\country{}}
\affiliation{$^4$\textit{Department of Computer Science, Delft University of Technology, The Netherlands}\country{}}
\affiliation{ywteng@iis.sinica.edu.tw; hwchen@arbor.ee.ntu.edu.tw; dnyang@iis.sinica.edu.tw\country{}} \affiliation{yvonneanne@dfinity.org; b05901081@ntu.edu.tw; lydiaychen@ieee.org\country{}}
\begin{abstract}
Online social networks have become a crucial medium to disseminate the latest political, commercial, and social information. Users with high visibility are often selected as seeds to spread information and affect their adoption in target groups. We study how gender differences and similarities can impact the information spreading process. Using a large-scale Instagram dataset and a small-scale Facebook dataset, we first conduct a multi-faceted analysis taking the interaction type, directionality and frequency into account. To this end, we explore a variety of existing and new single and multihop centrality measures. Our analysis unveils that males and females interact differently depending on the interaction types, e.g., likes or comments, and they feature different support and promotion patterns. We complement prior work showing that females do not reach top visibility (often referred to as the glass ceiling effect) jointly factoring in the connectivity and interaction intensity, both of which were previously mainly discussed independently.

Inspired by these observations, we propose a novel seeding framework, called \emph{\seeding}, which aims at maximizing spread while reaching a target user group, e.g., a certain percentage of females -- promoting the influence of under-represented groups. \seeding ranks influential users with two gender-aware measures, the Target HI-index and the Embedding index. Extensive simulations comparing \seeding with target-agnostic algorithms show that \seeding meets the target percentage while effectively maximizing the spread. \seeding can be generalized to counter different types of inequality, e.g., race, and proactively promote minorities in the society.  
\end{abstract}


\maketitle

\section{Introduction}\label{sec:introduction}
Online social networks are an integral part of people's lives in today's society. People share their experiences and opinions via text, audio and videos on social media, e.g., Facebook~\cite{Bakshy:WWW:2012:RoleSoc}, Twitter~\cite{Shirin:ICWSM:2016:TwitterGlassCeiling}, LinkedIn~\cite{Tang:HCI:2017:JobMarketGenderBias}, and Instagram~\cite{Stoica:www:2018:AlgoGlassCeiling}.
Social media grew from a platform to share personal experience into one of the main channels to disseminate information, e.g., news~\cite{Robert:JQ:1991:SocialNews}, political campaigns~\cite{Allcott:JEP:2017:SocialPolitics}, and product reviews~\cite{Dang:IS:2010:SocialProduct}. Influential users who are well connected in their social networks can critically accelerate the information spread and become the ideal seeds to affect others~\cite{Stoica:WWW:2020:SeedingDiver, Stoica:www:2018:AlgoGlassCeiling}.

Thus, it is of great interest to understand the characteristics of influential users, e.g., demographic traits and educational background. Several studies~\cite{Shirin:ICWSM:2016:TwitterGlassCeiling, aral2012identifying} point out that there is a correlation between the (perceived) gender of users and gaining visibility and influence on social media. Specifically, they have shown the existence of the~\emph{glass ceiling effect}, which makes it harder for females to become highly influential~\cite{Medel:SIGCSE:2017:EliminateGenderBias, Imtiaz:ICSE:2019:InvestigateGitHubGenderBias}, observed from their direct interactions with others, e.g., commenting or liking the posts. In other words, there are more males than females in the percentiles of the most popular users~\cite{Leavy:GE:2018:AI/MLGenderBias, Tang:HCI:2017:JobMarketGenderBias}.  These studies shed light  on social media usage patterns and show that gender disparity persists even for the younger generation of users~\cite{aral2012identifying}, by exhibiting the necessary conditions leading to the glass ceiling effect~\cite{Avin:ITCS:2015:HomophilyGlassCeiling, Shirin:ICWSM:2016:TwitterGlassCeiling, Cotter:SF:2001:GlassCeilingEffect}.

An important metric to quantify influence in prior art~\cite{aral2012identifying,Avin:PLOS:2018:ElitesSocialanetwork} is a node's degree, i.e., the number of neighbors a node is connected to under different types of interactions. Although such a measure demonstrates well that fewer females reach the highest tier of visibility compared to males~\cite{Avin:ITCS:2015:HomophilyGlassCeiling}, the intensity of the ties~\cite{Bakshy:WWW:2012:RoleSoc} among the users, e.g., the frequency of interaction among connected users, is often not considered. For instance, Stoica et al.~\cite{Stoica:www:2018:AlgoGlassCeiling} and Nilizadeh et al.~\cite{Shirin:ICWSM:2016:TwitterGlassCeiling} empirically show the existence of the glass ceiling on Instagram and Twitter, respectively, by focusing on the network structure or the interaction intensity separately. This leads us to the question that females still lag behind males in the highest percentiles when considering their degrees and interaction intensity.

Another essential question revolving around the glass ceiling concerns the endorsement process: who supports and promotes influential users. Gaining a deeper understanding of these patterns and structures may unveil solutions to break the glass ceiling. Parity and diversity seeding~\cite{Stoica:WWW:2020:SeedingDiver} are solutions proposed to maximize the information spread to a target gender group, e.g., females, by selecting seeds in a biased manner. Thus, females (forming an under-represented group among the most influential users) are promoted as seeds to disseminate information with the goal of increasing the reception of information among females. However, such approaches cannot explicitly accommodate promoting a specific ratio of the minority group. For instance, a scholarship program may want to guarantee a certain percentage of female award candidates and, thus, adequately inform the target audience. A seeding solution is urgently needed to achieve the target disparity despite different social network structures and interaction pattern manifestations.

We start by revisiting the glass ceiling effect on two popular platforms, i.e., \emph{Instagram} and \emph{Facebook}. In contrast to the prior art, we analyze each gender's visibility by both the degree and interaction intensity, quantifying the number of links and how often they are used, both from a sender and receiver perspective. We further differentiate between different interaction types, i.e., likes vs. comments (and tags for Facebook), to investigate if females face glass ceiling effects that hinder them from reaching higher visibility in casual social platforms. Our analysis initially focuses on the one-hop neighborhood in the social network, i.e., user pairs interacting directly with each other (single-hop analysis). Then, we propose a new centrality metric, termed \emph{HI-index}, to quantify the \emph{indirect network influence} by considering the multihop impact via friends of friends. Our detailed analysis demonstrates that there are number of criteria under which glass ceiling effects are observed, e.g., comment intensity and tag intensity in the single-hop analysis, and it points out unfavorable metrics with low visibility ranks in the overall network for females, e.g, the out-degree of comments. 

In the second part of this paper, we tackle the challenge of designing a seeding algorithm that can achieve a target \emph{disparity ratio} in groups under-represented in the highest visibility percentiles while ensuring maximal information spread. We derive a novel centrality measure, called \emph{Embedding index} via state of art graph neural networks. Leveraging our characterization study, we develop \seeding, which combines the proposed centrality measures, Target HI-index and Embedding index, and diffusion simulation in a novel way. The critical components of \seeding are the gender-aware ranking and seed selection methods. Our evaluation results show that \seeding can effectively achieve the target gender ratio and maximize the information spread compared to the baselines.
Our contributions are summarized as follows:
\vspace{-\topsep}
\begin{itemize}
    \item We conduct a novel gender gap analysis on social media that jointly factors in the connectivity and intensity from the dual perspectives of the sender and receiver sides. 
    \item We show where glass ceiling effects exist for different interaction types, directionality, degree, and link intensity. 
    \item We propose a novel centrality measure, HI-index, and apply it together with PageRank, to highlight and rank the influence of male/female users on the entire network.
    \item We develop a novel seeding framework, \seeding, that maximizes information spread and satisfies a target gender ratio that can overcome the disparity exhibited in the population and compare it to known approaches.
\end{itemize}

\vspace{-2mm}
\section{Related Work}\label{sec:related}

Prior art has extensively studied gender bias and its implications in the context of professional network and career paths~\cite{McPherson:ANNUREY.SOC:2001:BirdHomophily,Avin:ITCS:2015:HomophilyGlassCeiling}. In social network analyses, users' influence is commonly quantified by various centrality measures~\cite{Avin:PLOS:2018:ElitesSocialanetwork,Iyengar:MS:2012:ProductDiffusion,aral2012identifying}, often with an implicit assumption that the links are undirected and link quality is uniform. Hazard models~\cite{aral2012identifying,Iyengar:MS:2012:ProductDiffusion} are applied to estimate the information spread and social contagion. 
On the other hand, several studies stress the importance of the connection quality. 
Bakshy et. al~\cite{Bakshy:WWW:2012:RoleSoc} pointed out with a study on Facebook that users with high interaction intensity form stronger ties and are more influential than users with low intensity, and thus they are key in efficiently disseminating information. The increasing popularity of services to purchase likes and bot activities further raise an alarming concern on the connection quality and implications on social network analysis~\cite{Jonas::2017:masterThesis}.

\textbf{Gender Gap Analysis on Social Media.}
Several studies investigate the glass ceiling effect on different online social platforms. Typically they find that males achieve higher visibility and spread information faster~\cite{Shirin:ICWSM:2016:TwitterGlassCeiling,Stoica:www:2018:AlgoGlassCeiling,Kevin:SN:2008:facebooksurvey}, from the perspective of either the network connectivity or intensity.
\textit{\textbf{Twitter}}: Nilizadeh et al. ~\cite{Shirin:ICWSM:2016:TwitterGlassCeiling} show that the perceived gender affects user visibility in different metrics, e.g., the number of followers, retweets, and lists. From the \emph{complementary cumulative distribution function} (CCDF) of most of those measures, high percentile males achieve higher visibility than females. 
\textit{\textbf{Facebook}}: In~\cite{Kevin:SN:2008:facebooksurvey}, it leverages similarity calculations to quantify relationships with quadratic assignment procedure (QAP). More precisely, similarities under gender, race/ethnicity, and socioeconomic status (SES) are studied.
\textit{\textbf{Instagram}}: Stoica et al.~\cite{Stoica:www:2018:AlgoGlassCeiling} derived mathematical models to explain how recommendation algorithms reinforce the glass ceiling gap. Their study does not take the tie strength into account. 
By contrast, we analyze influential users by fusing both tie strength and degree via a novel centrality measures, HI-index, and observe different glass ceiling effects depending on the interaction types on Instagram and Facebook.

\textbf{Influence Maximization.}
Social influence on online platforms is an active research topic and one of the main focuses is on maximizing the influence ~\cite{Kempe2003}. The key challenge is to pick $k$ seeds which maximize the number of users receiving information spread by the selected seeds, which is proved as NP-hard under the Linear Threshold and Independent Cascade diffusion models~\cite{Kempe2003}. Recent studies \cite{Stoica:WWW:2020:SeedingDiver,gershtein2018balanced,fish2019gaps,ijcai2019-831,farnad2020unifying,ijcai2020-594} focus on reducing the gender gap while maximizing information spread. However, they cannot ensure a user-specified gender ratio in the influenced users. In this paper, \seeding takes a step further to maximize the overall information spread and achieve any required gender ratio by design.

\section{Glass Ceiling on Instagram via Visibility and Endorsement Analysis}\label{sec:instagram}

We conduct a static characterization of female and male interaction patterns, answering the questions if there exists a gender gap in terms of users' visibility and endorsement on Instagram when considering intensity and degree. The visibility analysis studies appreciation received by posts from different genders to investigate the evidence of the glass ceiling, i.e., if males are over-represented in higher percentiles. The endorsement analysis focuses on how users support posts from others. 
In contrast to \cite{Avin:ITCS:2015:HomophilyGlassCeiling, Stoica:www:2018:AlgoGlassCeiling, Shirin:ICWSM:2016:TwitterGlassCeiling, aral2012identifying}, we take both the number of interaction partners (\emph{degree}) and the number of interactions (\emph{intensity}) into account. Moreover, we consider two types of interactions separately for a more nuanced analysis, i.e., we distinguish likes from comments since the latter demands higher involvement and effort from the senders than the former. Our objective here is to revisit the glass ceiling from multi-faceted perspectives and search for insights to understand and design potential alleviation measures to reduce gender gaps on social media. 

We first introduce the Instagram dataset and then analyze the visibility and endorsement from users interacting directly with each other, i.e., single-hop analysis. Afterward, we apply centrality metrics that can factor in both interaction intensity and degree of the entire network, i.e., multi-hop analysis. To this end, we extend H-index~\cite{Alonso:JI:2009:Hindex} to propose \textit{HI-index} by considering interaction intensity to quantify the overall network visibility of male/female. Meanwhile, we also investigate the ranking produced by PageRank.

\begin{figure*}[t]
\centering
\begin{subfigure}{.20\linewidth}
     \centering
     \includegraphics[width=3.5cm,height=2.3cm]{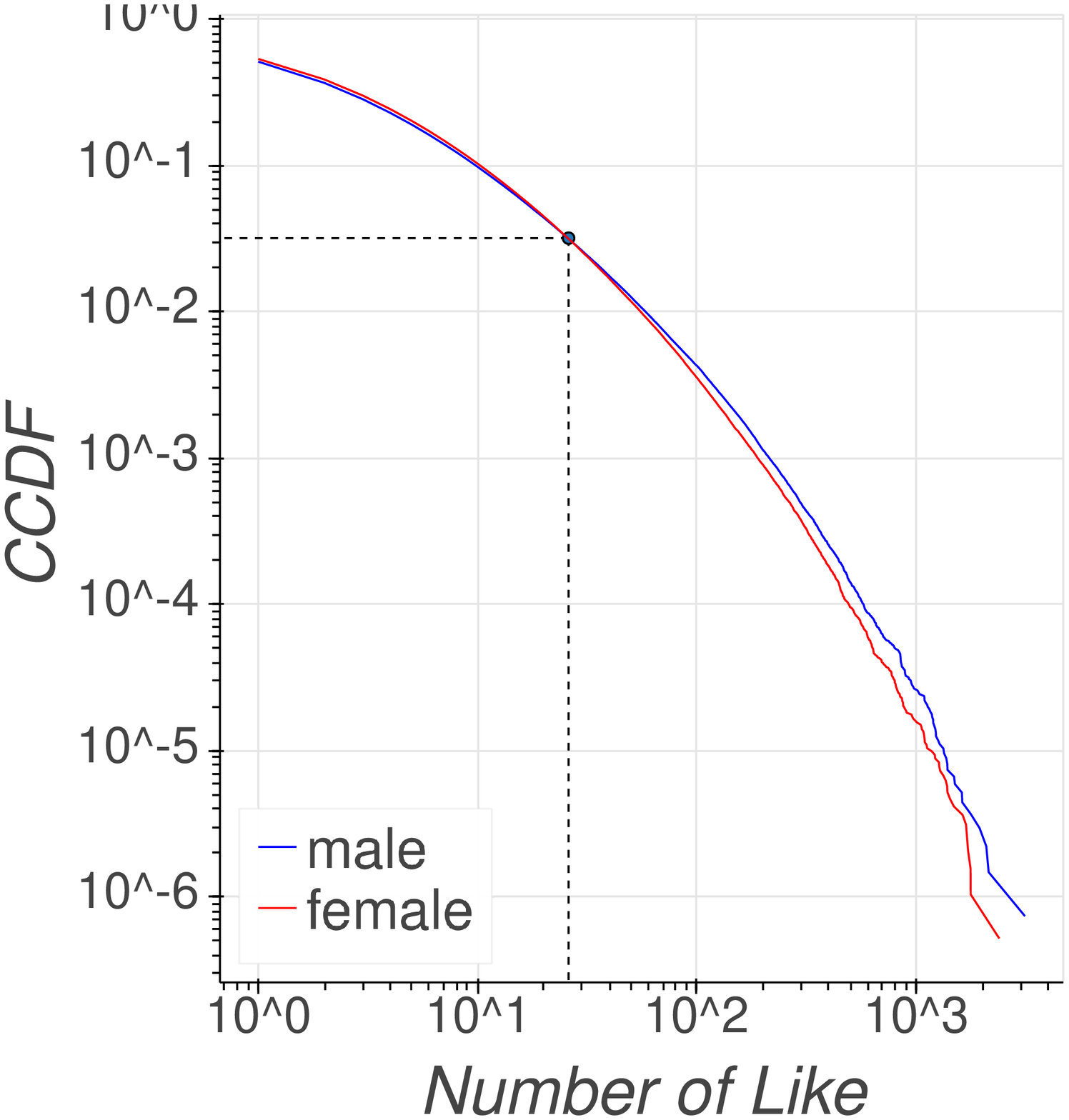}
     \caption{Received (Likes)}\label{fig:ccdflike_rec_r1u}
\end{subfigure}
     \hfill
\begin{subfigure}{.20\linewidth}
     \centering
     \includegraphics[width=3.5cm,height=2.3cm]{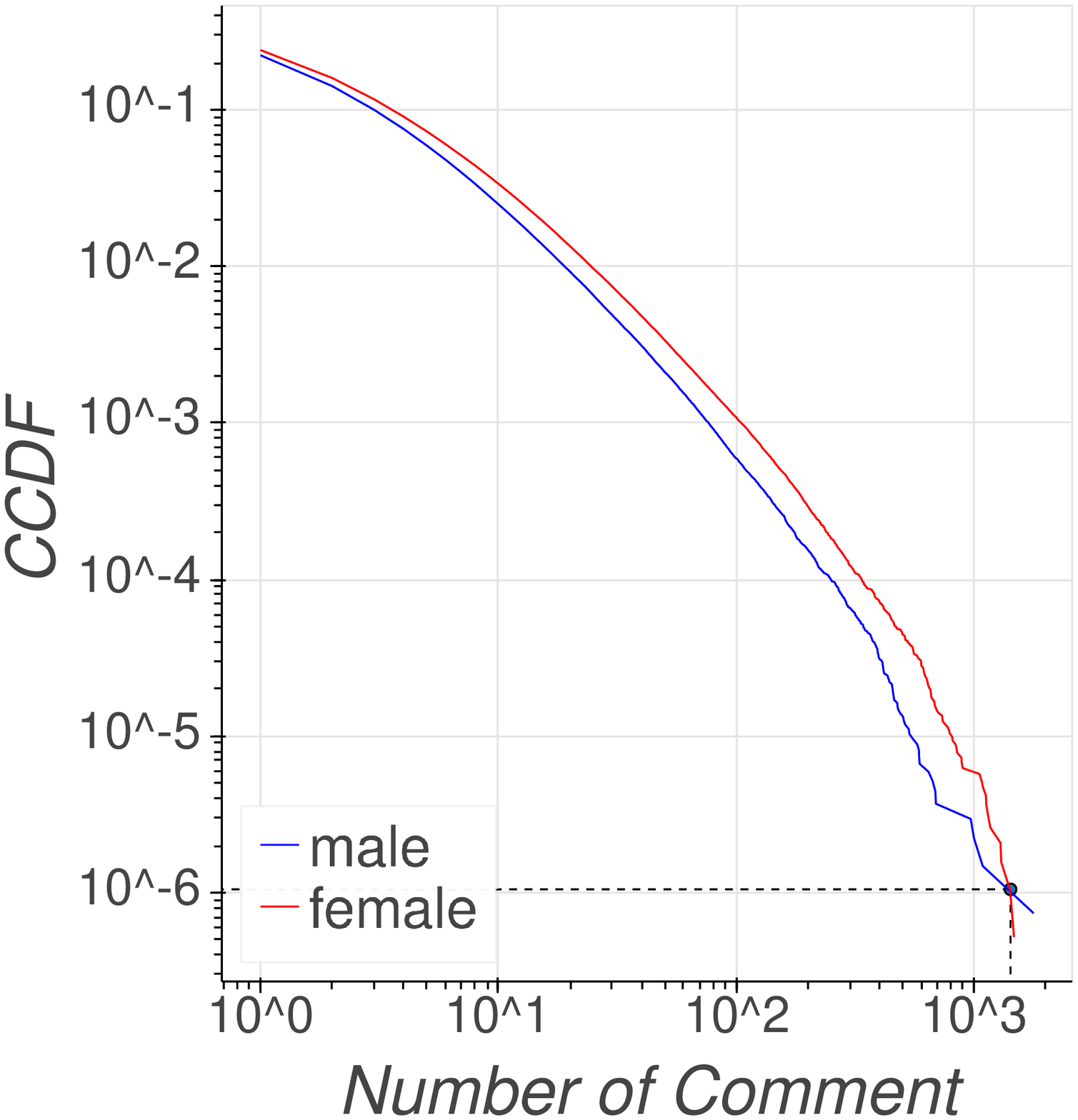}
     \caption{Received (Comments)}\label{fig:ccdfcomment_rec_r1u}
\end{subfigure}
    \hfill
 \begin{subfigure}{.20\linewidth}
     \centering
     \includegraphics[width=3.5cm,height=2.3cm]{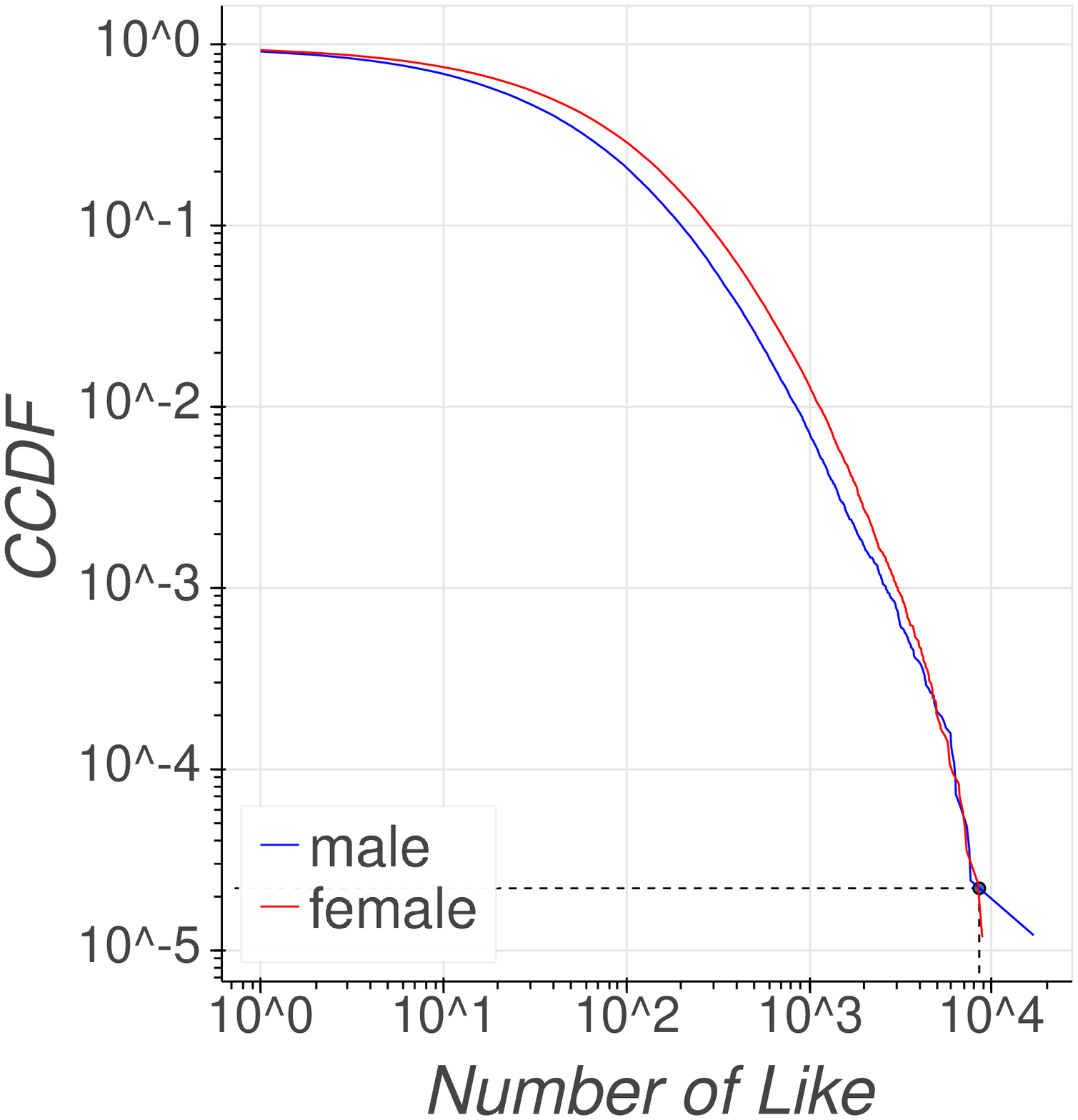}
     \caption{Sent (Likes)}\label{fig:ccdflike_sen_r1u}
 \end{subfigure}
 \hfill
 \begin{subfigure}{.20\linewidth}
     \centering
     \includegraphics[width=3.5cm,height=2.3cm]{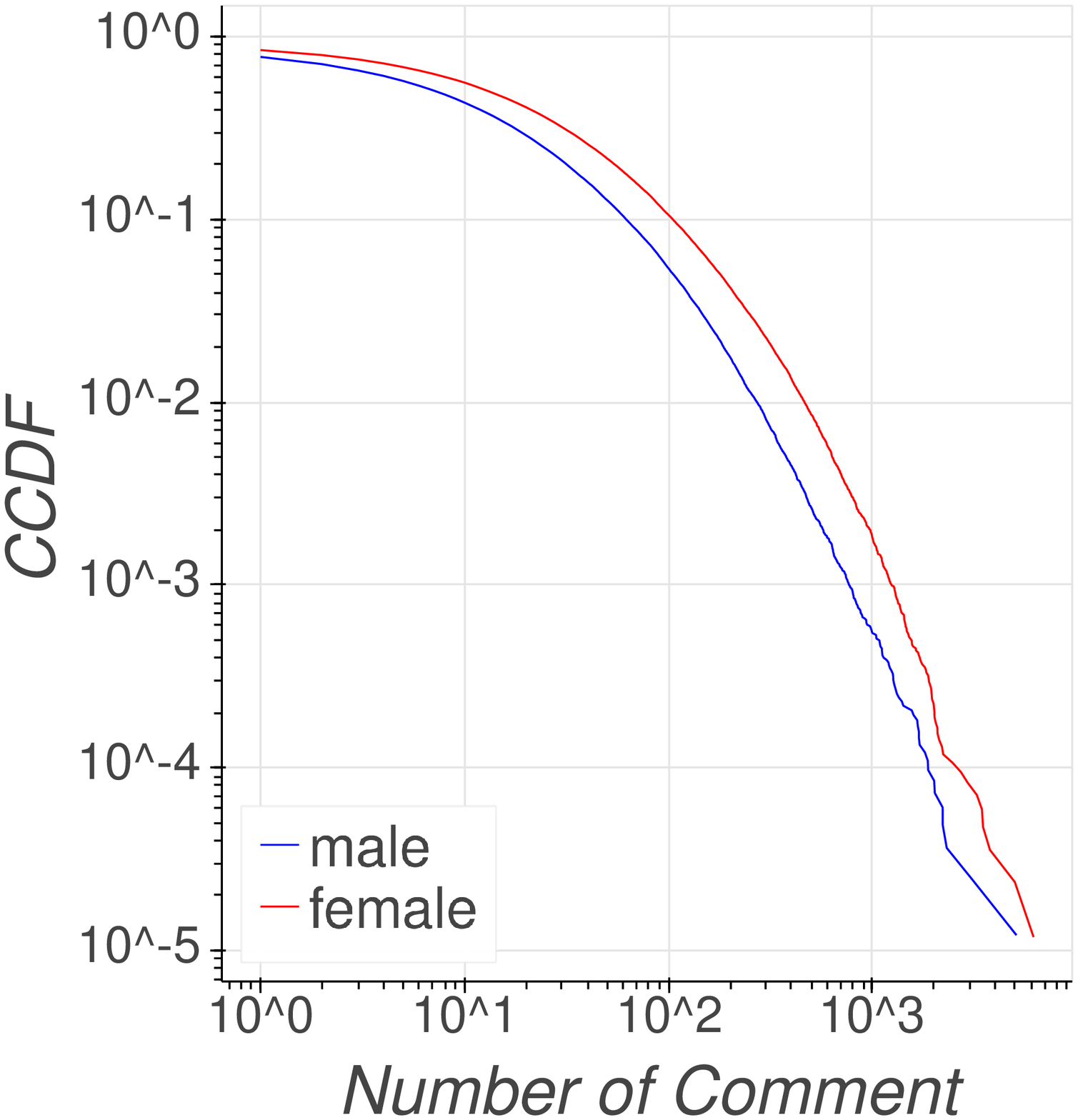}
     \caption{Sent (Comments)}\label{fig:ccdfcomment_sen_r1u}
 \end{subfigure}
\vspace{-.3cm}
\caption{Intensity of visibility (received) and endorsement (sent) on Instagram.} 
\vspace{-.3cm}
\label{fig:number_interactions_r1u}
\end{figure*}

\begin{figure*}[htp]
\centering
\begin{subfigure}{.20\linewidth}
    \centering
    \includegraphics[width=3.5cm,height=2.3cm]{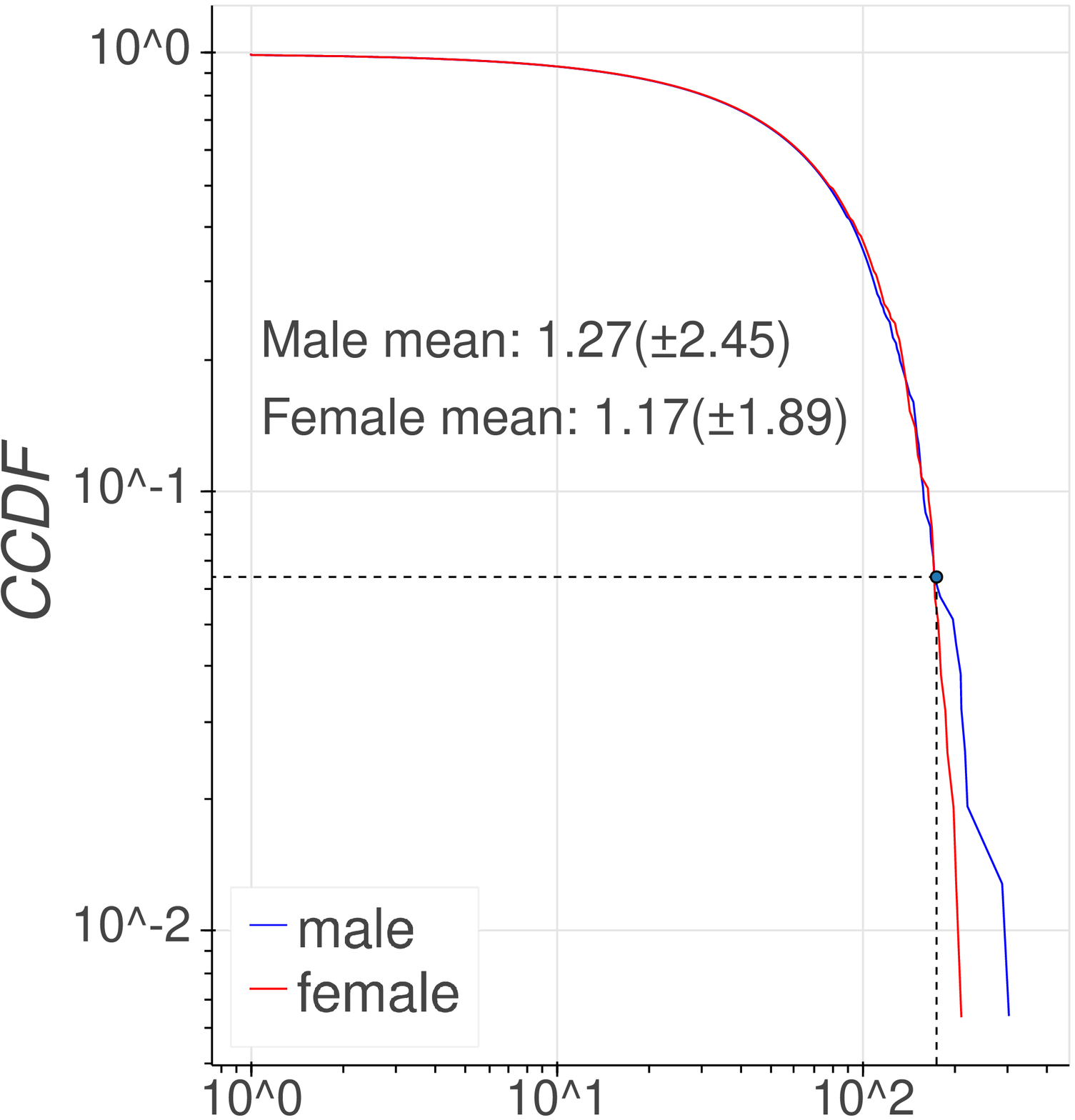}
    \caption{In-degree (Likes)}
    \label{fig:ccdflike_indegree_r1u}
\end{subfigure}
    \hfill
\begin{subfigure}{.20\linewidth}
    \centering
    \includegraphics[width=3.5cm,height=2.3cm]{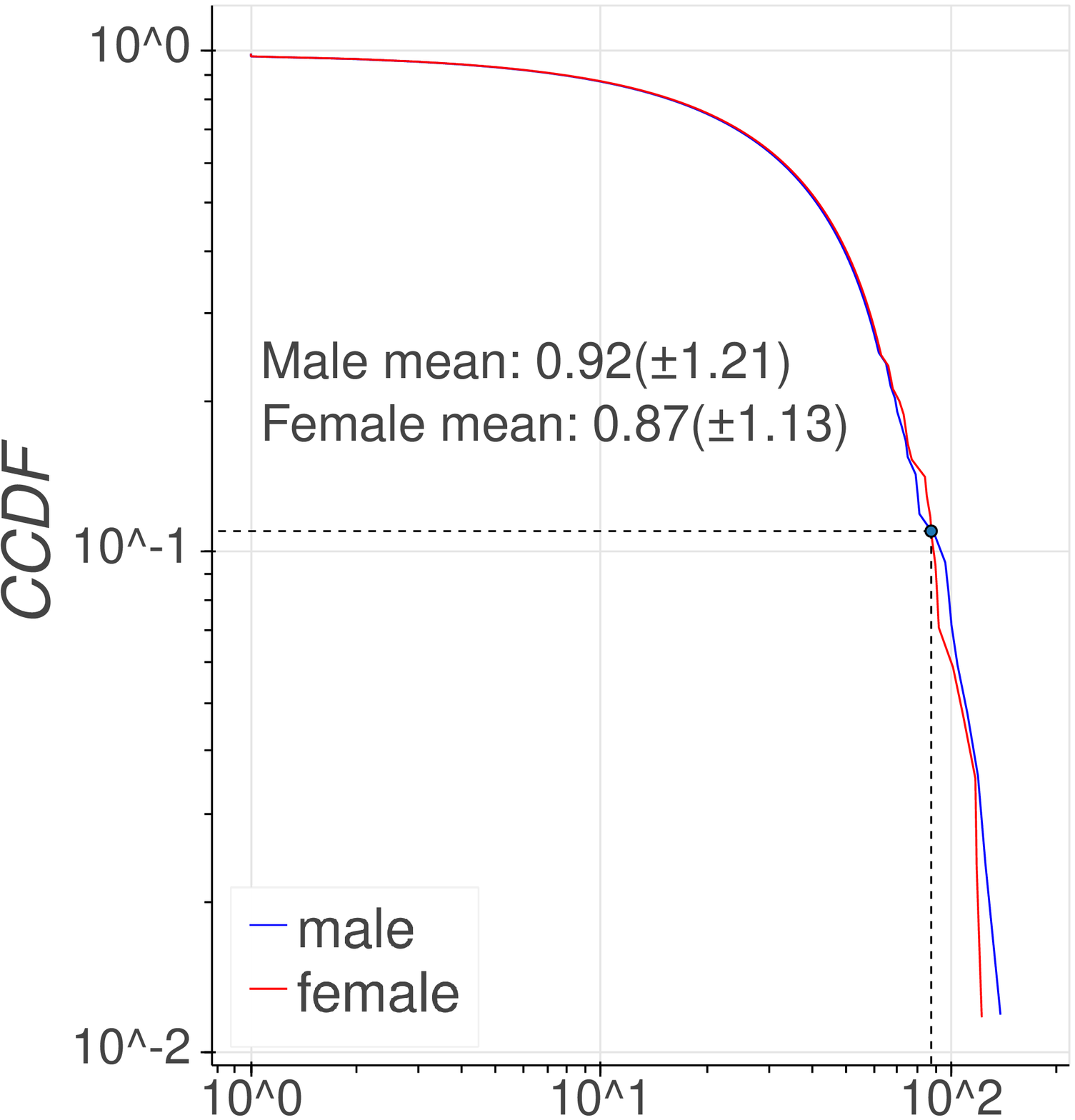}
    \caption{In-degree (Comments)}
    \label{fig:ccdfcomment_indegree_r1u}
\end{subfigure}
    \hfill
\begin{subfigure}{.20\linewidth}
    \centering
    \includegraphics[width=3.5cm,height=2.3cm]{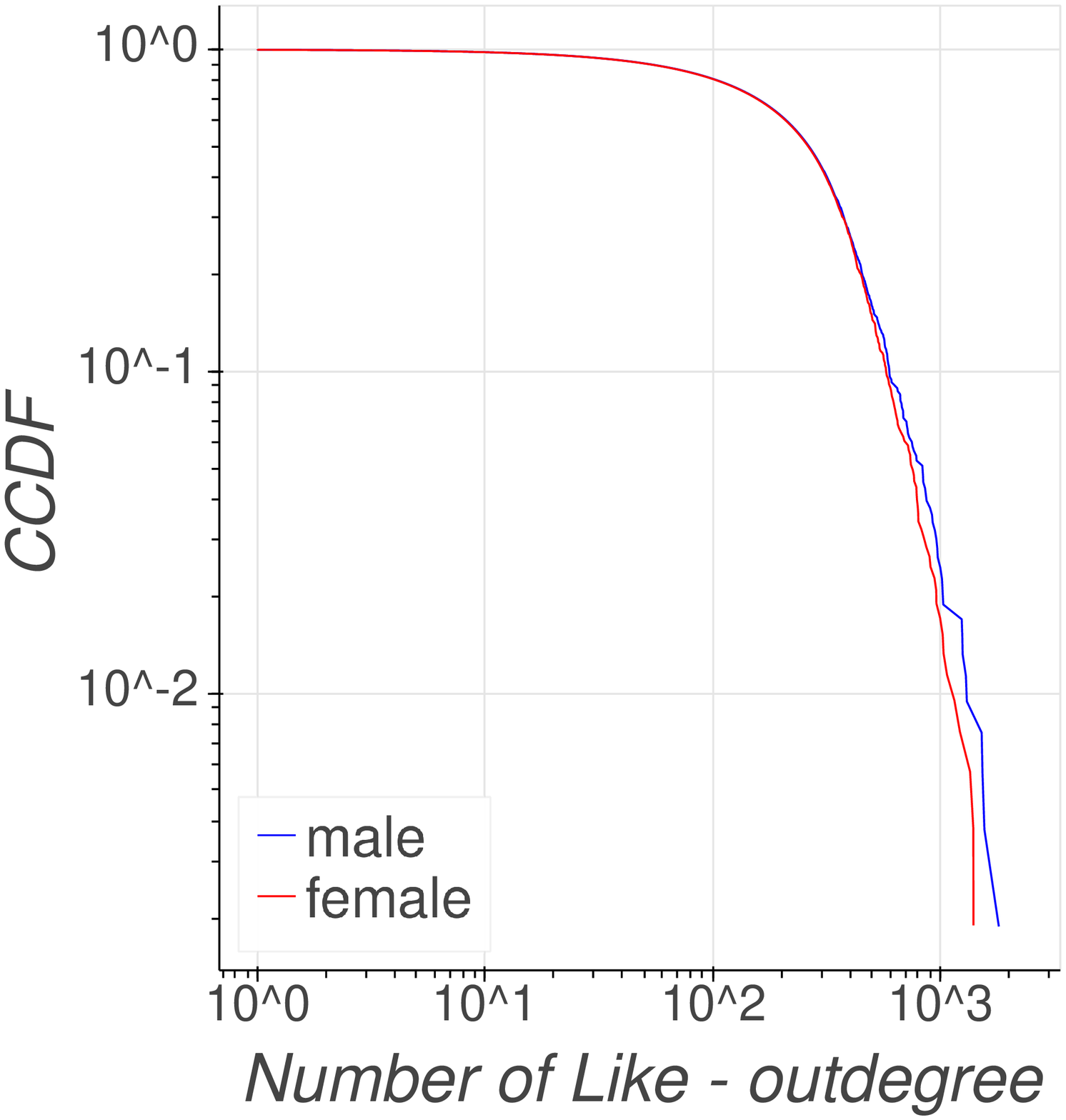}
    \caption{Out-degree (Likes)}\label{fig:ccdflike_outdegree_r1u}
\end{subfigure}
    \hfill
\begin{subfigure}{.20\linewidth}
    \centering
    \includegraphics[width=3.5cm,height=2.3cm]{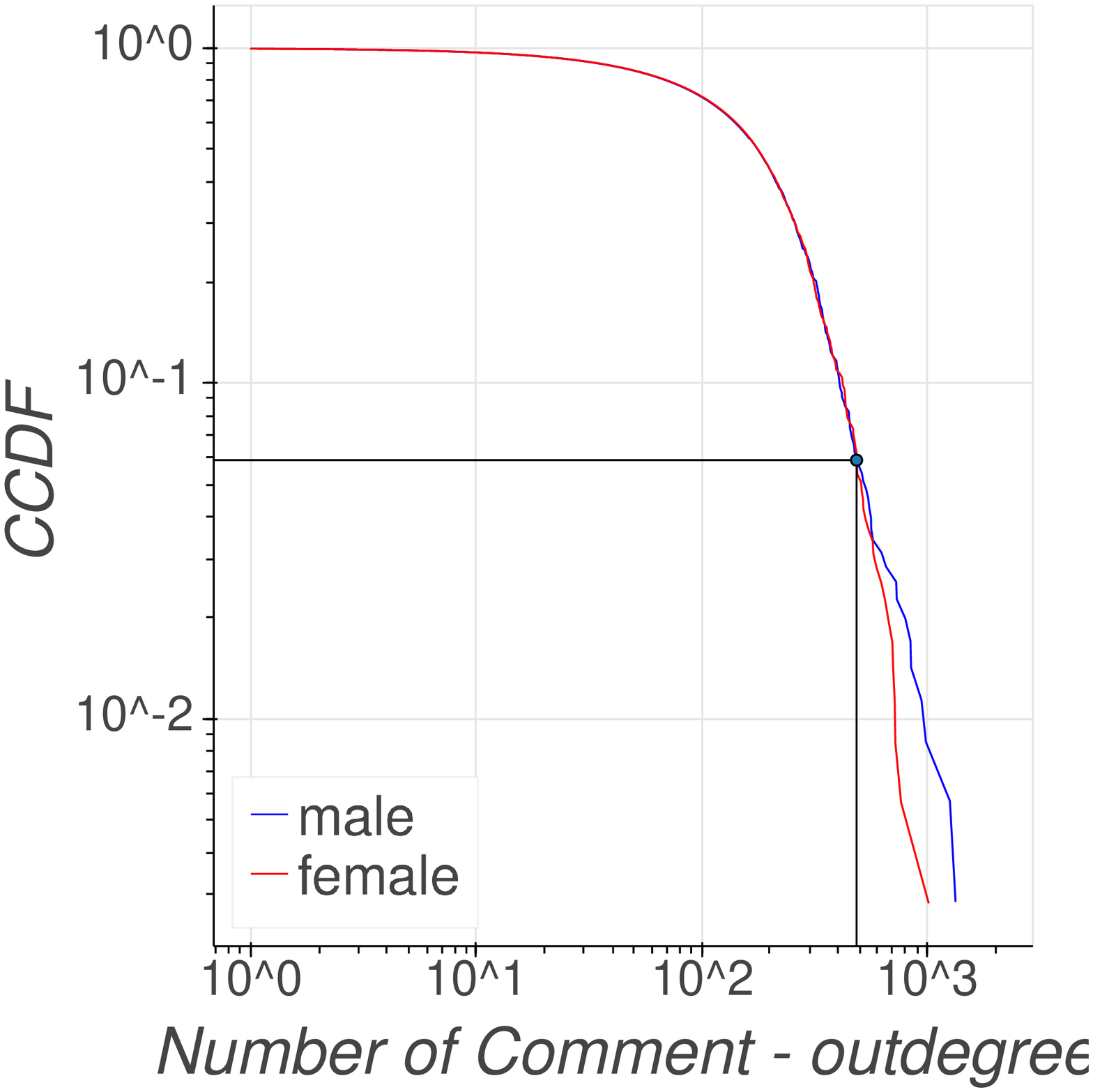}
    \caption{Out-degree (Comments)}\label{fig:ccdfcomment_outdegree_r1u}
\end{subfigure}
\vspace{-.3cm}
\caption{Degrees of visibility (received) and endorsement (sent) on Instagram. }
\vspace{-0.3cm}
\label{fig:in/out_degree_interactions_r1u}
\end{figure*}

\subsection{Instagram Dataset}
\label{apx:instagram}

\begin{table}[t]
    \centering
    \caption{Instagram gender distribution: senders vs. receivers.}
    \vspace{-2mm}
    \begin{footnotesize}
    \begin{tabular}{|l|ll|ll|} 
    \hline
            & \multicolumn{2}{c|}{Sender}                                    & \multicolumn{2}{c|}{Receiver}                                           \\ 
    \hline 
            & Male                           & Female                        & Male                             & Female                              \\ 
    \hline 
    Like    & 46.2\% (\textasciitilde{}39k) & 53.8\% (\textasciitilde{}46k) & 41.0\% (\textasciitilde{}1,143k) & 59.0\% (\textasciitilde{}1,648k)  \\
    Comment & 45.8\% (\textasciitilde{}38k)    & 54.2\% (\textasciitilde{}44k) & 40.7\% (\textasciitilde{}589k)   & 59.3\% (\textasciitilde{}857k)      \\
    Users   & 46.3\% (\textasciitilde{}40k)  & 53.7\% (\textasciitilde{}47k) & 41.3\% (\textasciitilde{}1,352k) & 58.7\% (\textasciitilde{}1,921k)    \\
    \hline
    \end{tabular}
    \end{footnotesize}
    \label{tab:basic}
\end{table}

\textbf{Collection method.} 
Stoica et al.~\cite{Stoica:www:2018:AlgoGlassCeiling} collected data of users' reactions with likes and comments on each other's posts over a $4$ year period. The data was retrieved with the Instagram API by first gathering a set of users, starting with the founder of Instagram and recursively retrieving the list of followers. Subsequently, data from the posters' perspective and how such posts receive likes and comments from other users was collected. For each interaction, the user ID, gender (derived from user names), responding type, i.e., likes or comments, and timestamps are recorded. Due to space and computation constraints, a maximum of $5$ interactions per post was sampled. Hence the interaction intensity is down-sized. The observation period in this paper spans from Jan 2015 to May 2016, starting right when the number of active Instagram users surged.

\textbf{Data characteristics.} The dataset comprises around $8$ million unique users with around $38$ million interactions among each other by considering both likes and comments before filtering. For both interaction types, we study the directed network formed by links representing users liking or commenting on another user's post.

\textbf{Prepossessing.} Since Instagram emerged to be one of the primary social platforms, especially for marketing campaigns, there is an increasing number of bot services~\cite{Jonas::2017:masterThesis} that offer to send ``likes'' for a payment. To avoid including such fake interactions in our analysis~\cite{Shirin:ICWSM:2016:TwitterGlassCeiling}, we filter out~\emph{inactive} users with only one interaction across our 17-month observation period. Roughly $57.45\%$ of users are inactive and hence removed from the following analysis. After the filtering, there are $1,375,637$ males and $1,499,145$ females. Table~\ref{tab:basic} summarizes the distribution of the dataset. 
There is a significantly higher number of receivers (roughly 30 times more) than senders who respond to posts via likes/comments due to the collection method, 
and more female senders/receivers than male ones. In terms of likes and comments, the number of likes observed is roughly $2.5$ times higher than the number of comments. The overall dataset is slightly unbalanced in terms of gender ratio. 

\textbf{Limitations.}
The data is collected starting from the founder's posts recursively (i.e., the receiver of comments/likes); therefore, the dataset exhibits a bias of this specific group of users, an unavoidable drawback when crawling information via such an API. Thus, the data may not cover the diversity of all Instagram users. 

\subsection{Single-Hop Analysis}

\begin{table}[t]
\small
\setlength{\tabcolsep}{0.4em}
    \caption{Mann-Whitney U test: the p-value on Instagram.}
    \label{T:u_test_ig}
    \vspace{-.3cm}
    \centering
    \begin{tabular}{|cc|ccc|}
    \hline & & Top 10\% & Top 1\% & Top 0.1\% \\\hline
    \parbox[t]{2mm}{\multirow{4}{*}{\rotatebox[origin=c]{90}{Comment}}} & Rec. intensity & 1.199e-156*** & 0.0*** & 4.285e-116*** \\
    & Sen. intensity & 1.996e-63*** & 0.003** & 0.102 \\
    & HI-index & 0.326 & 0.043* & 4.496e-09*** \\
    & PageRank & 0.0*** & 6.813e-66*** & 9.61e-10*** \\ \hline
    \parbox[t]{2mm}{\multirow{4}{*}{\rotatebox[origin=c]{90}{Like}}} & Rec. intensity & 0.0*** & 3.129e-82*** & 4.484e-33*** \\
    & Sen. intensity & 3.43e-48*** & 5.96e-05*** & 3.082e-4***\\
    & HI-index & 1.649e-212*** & 0.0*** & 3.517e-105*** \\
    & PageRank & 1.347e-44*** & 1.416e-163*** & 2.493e-26*** \\ \hline
    \multicolumn{5}{l}{*, **, and *** denote $p\leq 0.05$, $p \leq 0.01$, and $p\leq0.001$, respectively.} \\
\end{tabular}
\end{table}

\subsubsection{Interaction Intensity}

Using the intensity of comments and likes, we aim to answer which gender is more \emph{influential} and which gender is more likely to \emph{endorse} others. Note that, typically, nodes' influence is derived from the degree, often with an implicit assumption of equal link quality~\cite{Dong:SIGKDD:2017:StructuralHomophily, Stoica:www:2018:AlgoGlassCeiling, Yang:NAS:2012:leadershipSuccess, Avin:PLOS:2018:ElitesSocialanetwork}, which is quite imprecise.

\textbf{Visibility Intensity}.
Figures ~\ref{fig:ccdflike_rec_r1u} and~\ref{fig:ccdfcomment_rec_r1u} summarize the total number of likes and comments \textit{received} by unique female/male users, respectively. We plot the complementary cumulative distribution function (CCDF) for females and males separately. The tail of the CCDF represents the most popular users, i.e., those who received the most likes/comments (referred to as the \textit{top-ranked} users) in their respective gender. By comparing the tails of female and male CCDFs, we can see if highly visible females receive as many likes/comments as highly visible males. Following~\cite{Stoica:www:2018:AlgoGlassCeiling}, Table~\ref{T:u_test_ig} shows the statistical confidence for top-ranked user groups (Top $10\%$, $1\%$, and $0.1\%$).\footnote{Due to space constraints, please refer ~\cite{online} for more details.}

We observe that much more likes are received than comments, which is not surprising as the effort to write a comment is significantly higher than clicking a like button~\cite{Dongyan:NC:2017:LikeComment}. Specifically, the mean number of likes received per user ($5.16$) is around three times higher than that of comments received per user ($1.79$). 
In Figure~\ref{fig:ccdflike_rec_r1u}, the higher visibility values for males at the tail indicate that top-ranked males receive more likes than top-ranked females, though both male and female users receive on average $5.16$ likes. This observation fits with the glass ceiling effect pointed out in ~\cite{Avin:ITCS:2015:HomophilyGlassCeiling,Stoica:www:2018:AlgoGlassCeiling}, i.e., females fall behind males in top-ranked positions but not in the lower-ranked positions. Specifically, such a cross point happens around the top $3.2\%$ for likes. 
However, such a glass ceiling effect is not observed for comments, shown in Figure~\ref{fig:ccdfcomment_rec_r1u}. Females constantly receive more comments than males at any rank. Consequently, the average number of comments received by females is $1.95$, around $23\%$ higher than for males ($1.59$). Even though a crossing exists in the highest percentile, only four males receive more comments than females, which should probably be considered as outliers.

\textbf{Endorsement Intensity}.
Displaying the endorsement activities, Figures~\ref{fig:ccdflike_sen_r1u} and ~\ref{fig:ccdfcomment_sen_r1u} summarize the number of likes and comments \textit{sent} by unique female/male users, respectively. We plot the CCDFs for females and males separately. The tails of these CCDFs illustrate how highly active female/male users endorse others. 

The intensity of sending likes is almost three times higher than commenting, matching the CCDF of the receiver side. However, the shapes of the CCDFs are quite different: the CCDFs of likes/comments sent are higher than the CCDFs computed for the receiver side. This observation shows a mass of receivers having low visibility, but senders are more evenly distributed in terms of their endorsement efforts, i.e., the endorsement intensity varies considerably. 
We observe that females are more active than males regarding both likes and comments, as shown in Figures~\ref{fig:ccdflike_sen_r1u} and~\ref{fig:ccdfcomment_sen_r1u}, respectively. On average, females send $117$ likes (around $39\%$ more than males) and $44$ comments (around $69\%$ more than males). Note that the crossing in the highest percentile in Figure~\ref{fig:ccdflike_sen_r1u} is caused by a tiny number ($3$ users) of males (similar to Figure~\ref{fig:ccdfcomment_rec_r1u}), which should also be regarded as outliers. We notice that different interaction types in the sender point-of-view are both led by females.

 \begin{shaded*}
     \noindent\textit{\textbf{Takeaway}: Females are much more active in giving comments and likes than males, at almost all ranks. The variability of endorsement intensity is lower than the visibility intensity.}
 \end{shaded*}

\vspace{-.3cm}
\subsubsection{Interaction Degree}

Different from the previous subsection, we resort to the interaction degree to answer the gender difference in achieving high visibility and active endorsement. 

\textbf{Visibility Degree}.
Figures~\ref{fig:ccdflike_indegree_r1u} and ~\ref{fig:ccdfcomment_indegree_r1u} summarize the numbers of unique users from whom a particular user receive likes or comments, i.e., in-degree. We separate males and females and plot the CCDFs of in-degree for likes and comments. Different from \cite{Shirin:ICWSM:2016:TwitterGlassCeiling}, the CCDFs of in-degrees do not show a strong power-law behavior, i.e., a significant fraction of users account for a broad spectrum of degree, which can be regarded as a data-dependence characteristic.

The average degrees of likes and comments are $1.21$ and $0.89$, respectively. Although the difference between the average like and comment intensity is around a factor three, the resulting degree difference through these two types of interaction is lower. This can be explained by the frequency and repetitive interactions over certain user pairs, indicating a stronger tie. 
Figure~\ref{fig:ccdflike_indegree_r1u} again shows a glass ceiling effect: females attain higher degrees than males for low and medium percentiles. However, among the $6.4\%$ top-ranked males ($113k$ individuals) have consistently higher degrees than their female counterparts. 
As for the degree established through comments (Figure~\ref{fig:ccdfcomment_indegree_r1u}), there is no visible difference between males and females below the top-ranked $11\%$ ($9k$ users) and a minor gender gap for higher ranked users. Recall that females receive significantly higher comment intensity than males. Such a discrepancy can be possibly explained by how such comments are distributed across different senders. Females appear to receive comments from a smaller group of users with higher intensity, whereas males receive comments from a larger group in lower intensity.

\begin{shaded*}
     \noindent\textit{\textbf{Takeaway}: For likes and comments, a clear glass ceiling effect can be observed: top-ranked males receive visibility from a larger set of users than females.}
 \end{shaded*}

\textbf{Endorsement Degree.}
Figures~\ref{fig:ccdflike_outdegree_r1u} and ~\ref{fig:ccdfcomment_outdegree_r1u} summarize how many unique users are supported by a particular female or male user, i.e., out-degree, in CCDF plots. 

As the total number of senders is much lower than that of receivers, the average out-degrees, $45.21$ and $21.27$ for likes and comments, respectively, are much higher than the average in-degrees per user. 
With the above observations, we have valuable insights summarized as follows. In the intensity analysis, males perceive higher visibility with likes and females with comments. In the endorsement aspect, females are more active regarding both interactions. However, recall that in the degree analysis, males dominate in all scenarios (in/out-degree on like/comment). This implies that females prefer to interact with a smaller group of users than males do, which means that males are part of larger groups to which they send or from which they receive likes and comments.

\begin{shaded*}
     \noindent\textit{\textbf{Takeaway}: Top-ranked males support a larger group of users with lower intensity, whereas females support a small group of users with higher intensity.}
 \end{shaded*}

\subsection{Multi-hop Analysis} \label{sec:multihop}
The strength of the gender differences varies across the two visibility measures, i.e., degree and intensity, under the two interaction types, likes and comments. In this section, we take a complementary perspective to investigate influence, beyond the one-hop neighborhood. Our objective here is to combine both the degree and intensity simultaneously for direct and indirect social ties, i.e., friends of friends. To this end, we propose to use two centralities to quantify nodes' influence, a novel centrality \emph{HI-index}, inspired by the H-index ranking of influential authors~\cite{Alonso:JI:2009:Hindex}, and PageRank~\cite{Page:SI:1999:PageRank}. 

\begin{figure}
\centering
\hfill
\begin{subfigure}{.40\linewidth}
    \centering
    \includegraphics[width=3.5cm,height=2.6cm]{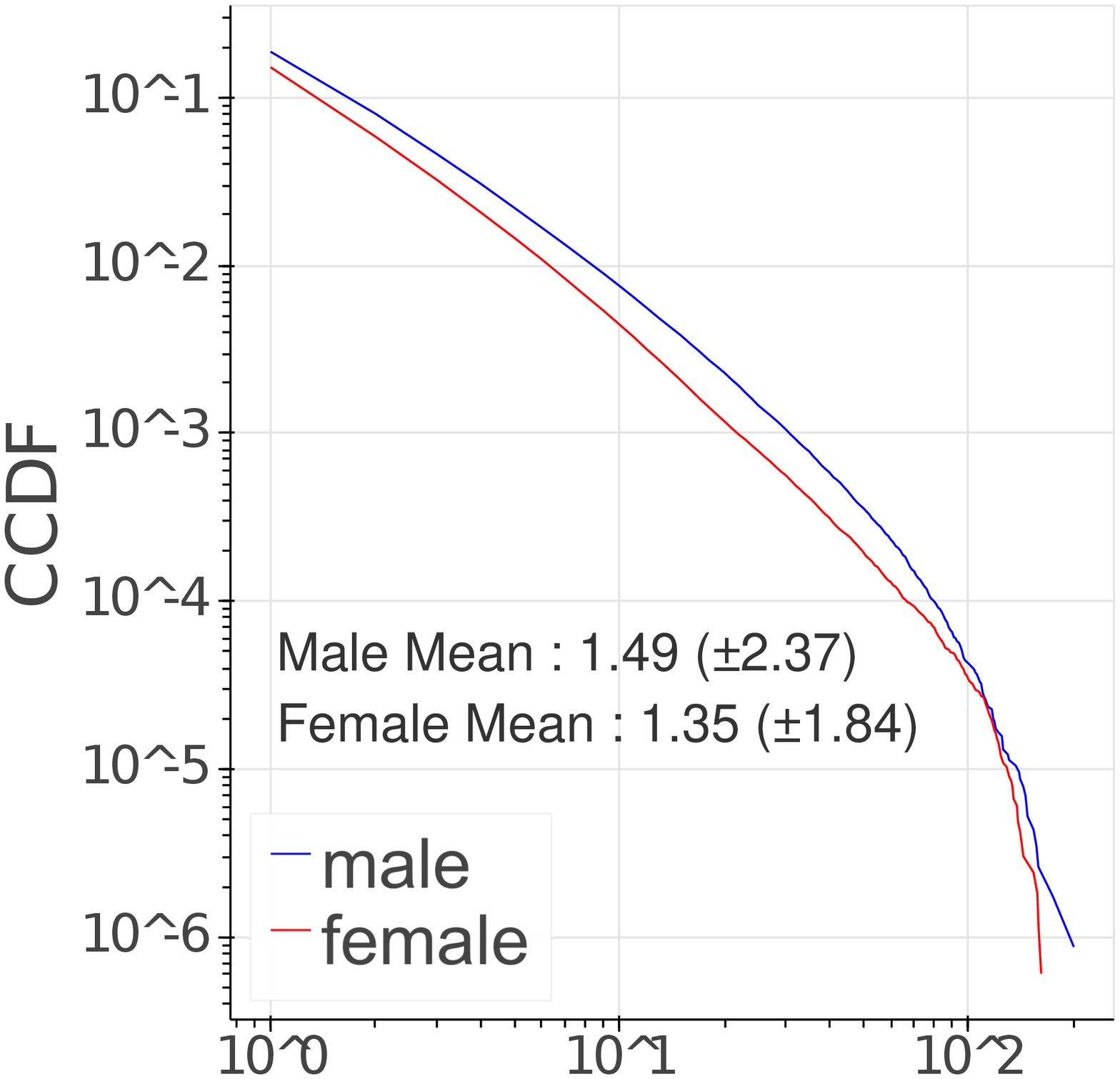}
    \caption{HI-index (Likes)}
    \label{fig:ccdflike_hindex_rec_r1u}
\end{subfigure}
    \hfill
\begin{subfigure}{.40\linewidth}
    \centering
    \includegraphics[width=3.5cm,height=2.6cm]{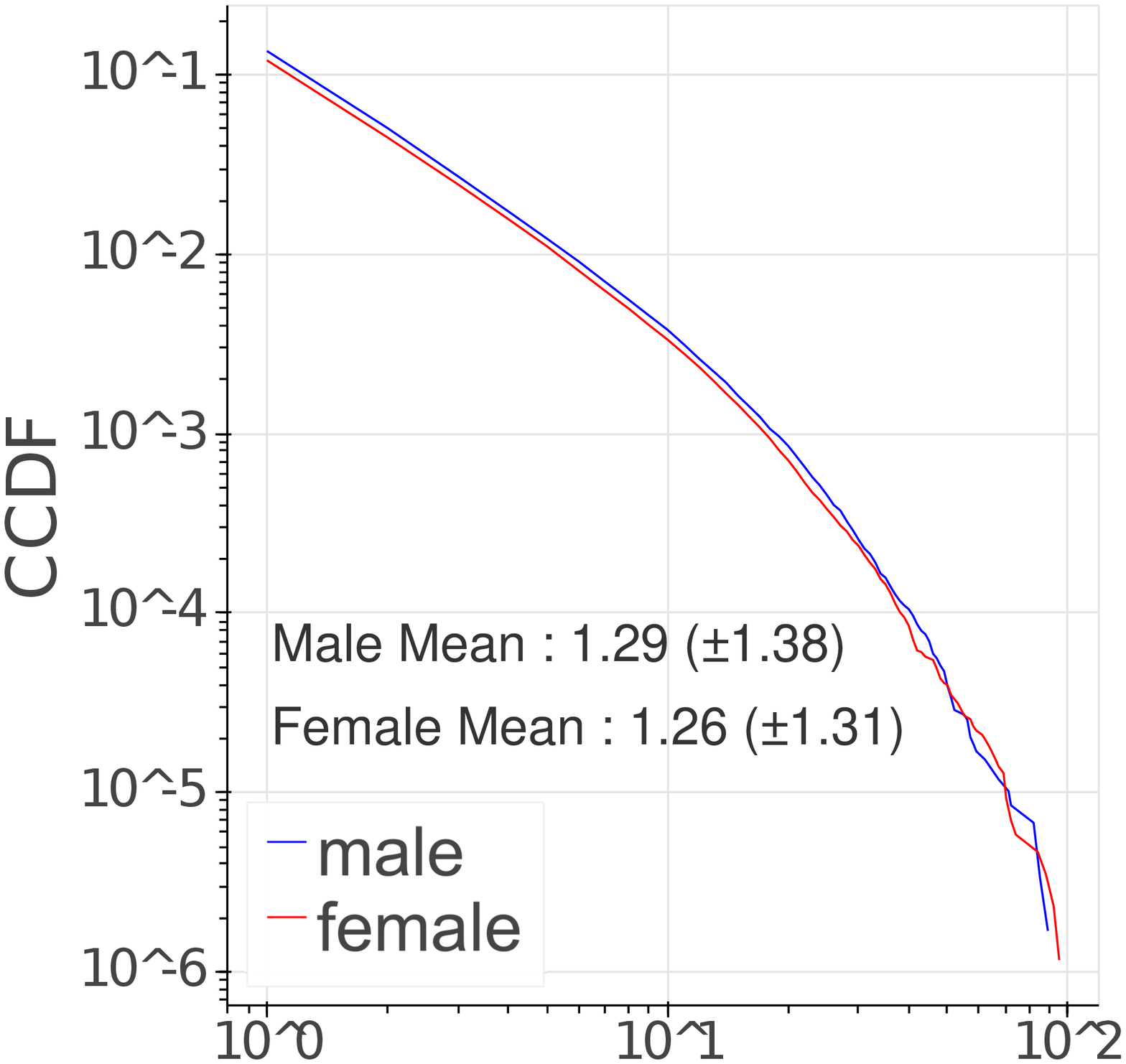}
    \caption{HI-index (Comments)}
    \label{fig:ccdfcomment_hindex_rec_r1u}
\end{subfigure}
\hfill
\vspace{-.3cm}
\caption{HI-index of visibility (received) on Instagram.} 
\label{fig:h_index_interactions_r1u}
\end{figure}

\subsubsection{HI-index}
The definition of the HI-index is derived from the H-index~\cite{Alonso:JI:2009:Hindex}, which attributes an author's productivity and citation impact by the number of publications that have been cited at least a certain number of times. In other words, the H-index relies on the degree of publications in the network formed by citations. We extend it to consider the interaction intensity as follows.

\begin{definition} The
  \emph{HI-index} of a user $v_i$ is defined as the maximum number $H$ such that $v_i$ has at least $H$ neighbors who interact with $v_i$ and any other users in the network at least $H$ times. 
  Let $N(v_i,n)$ denote the number of $v_i$'s one-hop neighbors who interact with others at least $n$ times.
  We can formulate the HI-index of $v_i$ as
  \begin{equation}\small
    H(v_i)=\max _{n} \min_{n \in \mathbb{I}^+} (N(v_i,n), n).
  \end{equation}
\end{definition}

The HI-index goes beyond a single hop analysis by evaluating all interactions of neighbors, and considers more than the interactions between the author of a post and its supporters. Thus the two-hop neighborhood of a user determines its HI-index.

Figure~\ref{fig:h_index_interactions_r1u} summarizes the CCDFs of the HI-index for males and females for likes and comments. The tails of the CCDFs represent the highly ranked users who have large HI-index values. The average HI-index values for likes and comments are $1.41$ and $1.28$, respectively.
We observe that the HI-index values of males exceed those of females at the same percentiles in Figure~\ref{fig:ccdflike_hindex_rec_r1u}. 
In other words, males receive higher visibility than females directly and indirectly in the social network. On the other hand, in Figure~\ref{fig:ccdfcomment_hindex_rec_r1u}, the male and female CCDFs are much closer to each other regarding comments, showing no remarkable difference. Recall the single-hop analysis in Section~\ref{sec:instagram} that females attain higher centrality values for comments. Females only attain similar visibility as males when considering their direct and indirect neighbors. Females' HI-index is generally lower than males because of the dual emphasis of interaction intensity and degree. Remember from the previous analysis that females tend to establish higher intensity interactions but in lower quantities than males. As the HI-index considers the interaction in both the single-hop and two-hop neighborhoods, females reach lower HI-index values than visibility measures using intensity only.

\begin{shaded*}
     \noindent\textit{\textbf{Takeaway}: When combining the number of direct and indirect interactions and their intensity in the HI-index, males achieve higher or equivalent visibility compared to females (i.e., the blue line is higher or comparable to the red line in Figure~\ref{fig:h_index_interactions_r1u}). }
 \end{shaded*}

\subsubsection{PageRank}

Another widely adopted metric to quantify the influence over an entire network is the PageRank centrality~\cite{Page:SI:1999:PageRank}, devised initially to sort web pages by their popularity. A page's popularity is measured by the number of times it is linked to by other pages, weighed by its popularity in turn. After that, PageRank recursively computes the steady-state probability of being at a page when following links at random.\footnote{We mainly exploit PageRank for analysis due to the scalability on massive networks, instead of other centrality measures, e.g., betweenness and closeness~\cite{Page:SI:1999:PageRank}.}

\begin{figure}[t]
\centering
\hfill
\begin{subfigure}{.40\linewidth}
    \centering
    \includegraphics[width=3.5cm,height=3.2cm]{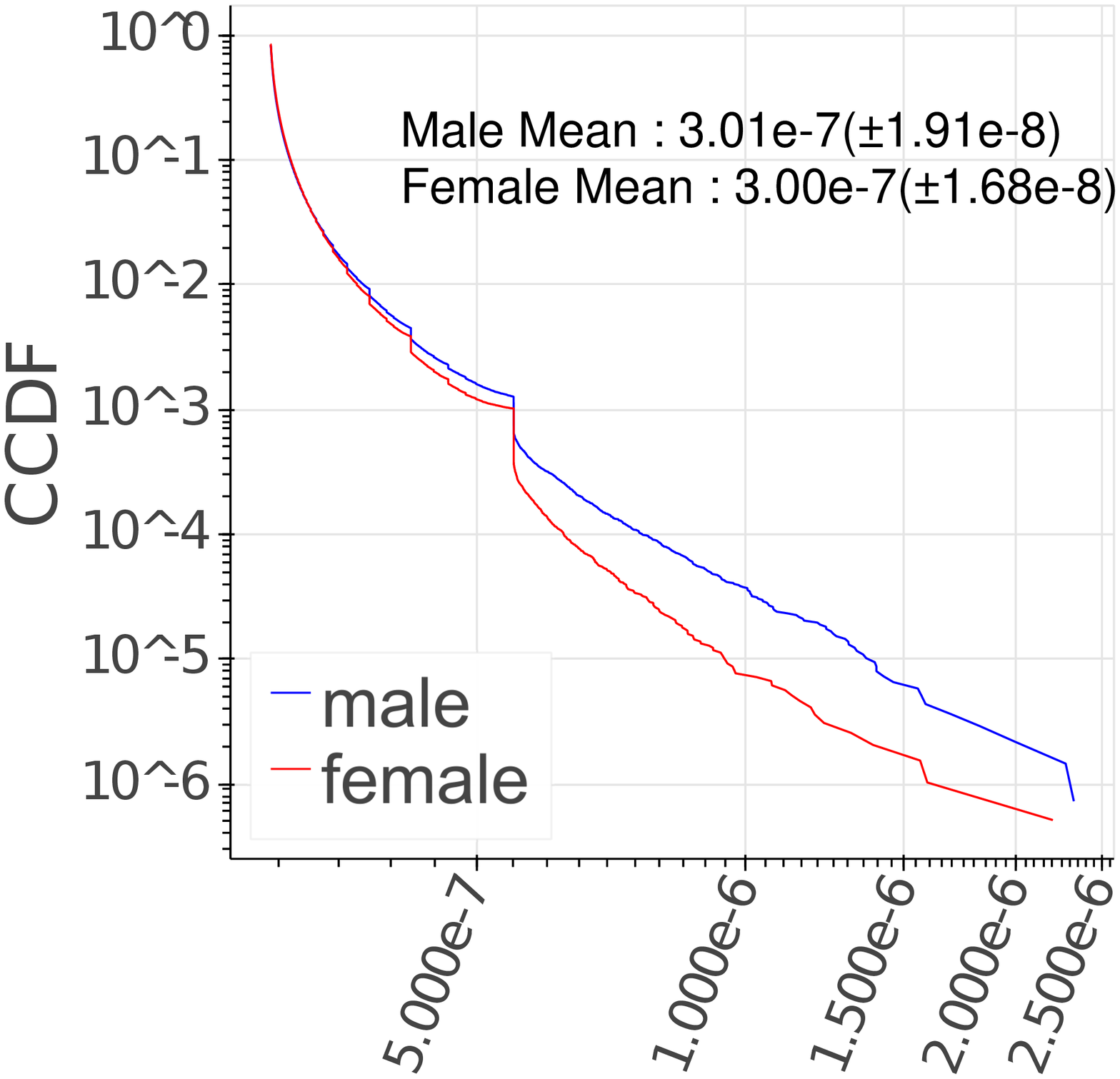}
    \caption{PageRank (Likes)}
    \label{fig:ccdflike_PageRank_r1u}
\end{subfigure}
    \hfill
\begin{subfigure}{.40\linewidth}
    \centering
    \includegraphics[width=3.5cm,height=3.2cm]{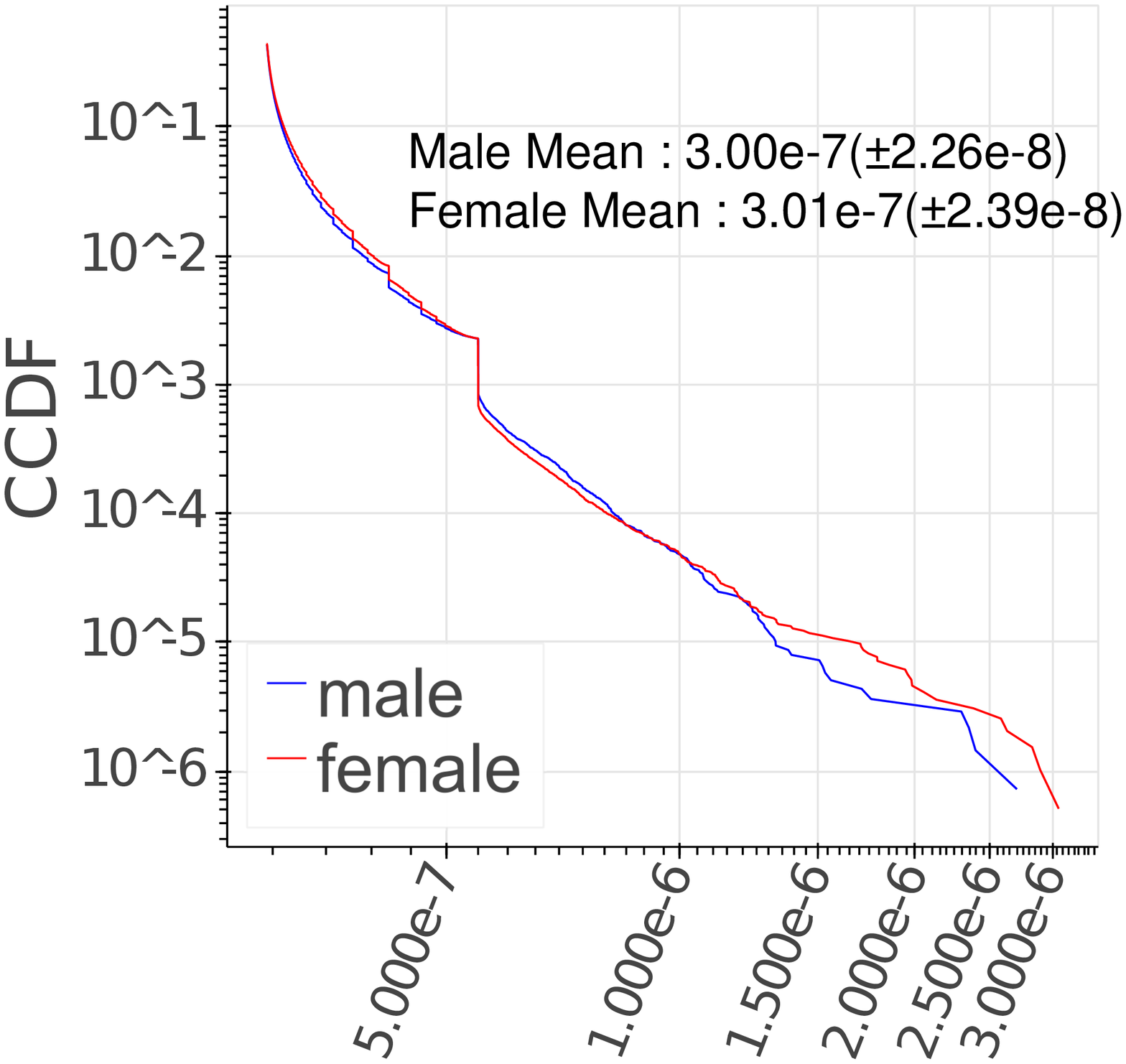}
    \caption{PageRank (Comments)}
    \label{fig:ccdfcomment_PageRank_r1u}
\end{subfigure}
\hfill
\vspace{-.3cm}
\caption{PageRank of visibility (received) on Instagram.} 
\label{fig:pagerank_interactions_r1u}
\end{figure}

Both HI-index and PageRank measure the visibility level of nodes taking the degree and the intensity of interactions into account. However, there are differences. In HI-index, for a user to achieve a high value, a sufficient number of highly visible neighbors is necessary. In other words, HI-index requires to exceed two thresholds, interaction partner quality and quantity. In PageRank, a user may be highly visible, even if the number of neighbors is not very high if, in turn, their neighbors are highly visible.

We summarize the PageRank results for likes and comments in Figures~\ref{fig:ccdflike_PageRank_r1u} and ~\ref{fig:ccdfcomment_PageRank_r1u}, respectively.  
In Figure~\ref{fig:ccdflike_PageRank_r1u}, one can observe that females attain higher PageRank visibility than males from the starting up to $5\%$ percentile. However, medium and top-ranked males outperform females significantly, shown by their higher CCDF. 
In contrast to HI-index, top-ranked females attain higher PageRank visibility than top-ranked males in Figure~\ref{fig:ccdfcomment_PageRank_r1u}. This is because PageRank considers the visibility of neighbors without any threshold of degree, which leads to that the females with low degrees but high intensity can be ranked high in terms of PageRank.
Figure~\ref{fig:ccdfcomment_PageRank_r1u} shows that females' PageRank visibility for comments is dampened compared to the pure comment intensity (Figure~\ref{fig:ccdfcomment_rec_r1u}) yet the top-ranked females receive higher visibility than males due to the higher intensity females receive from their supporters.

\begin{shaded*}
     \noindent\textit{\textbf{Takeaway}: When measuring the visibility by the PageRank centrality, females face glass ceiling effects in the like network, while they reach higher visibility for comments. }
 \end{shaded*}

\section{Glass Ceiling on Facebook}

In this section, we first introduce the second dataset containing user interactions on Facebook. Then, we will focus more on the intensity of received interactions and a new interaction type (tag) to complement the Instagram analysis.
\subsection{Facebook Dataset}
\label{apx:fb}
\textbf{Collection method.} Using the Facebook API, we collected data from users who study at $25$ university departments. The users comprise $1870$ voluntary senior students of the before-mentioned departments and all interactions between them are retrieved. For each interaction, the user ID, gender (derived by questionnaire), interaction type, and timestamps are recorded. In addition, we also collected user profiles, e.g., academic standing and hometown, by questionnaires. The period of iterations spans from March 2008 to May 2016, and $97.26\%$ of interactions are after August 2012. 

\textbf{Data characteristics.} The dataset comprises around 20 million interactions by $1870$ unique users ($765$ males and $1,105$
females). 
In addition to likes and comments, Facebook supports a third type of interactions where users can refer to each other via \emph{tags} in posts. Different from the Instagram dataset, we keep all interactions without any filtering.
Table~\ref{tab:basic_fb} summarizes the main statistics of how males/females interact with each other through likes, comments and tags. Note that the percentages on the receiver side are very similar to those observed on Instagram, while we see a higher female participation among the sender numbers in this dataset.

\begin{table}[t]
    \centering
    \caption{Facebook gender distribution: senders vs. receivers.}
    \vspace{-2mm}
    \begin{footnotesize}
    \begin{tabular}{|l|ll|ll|} 
    \hline
            & \multicolumn{2}{c|}{Sender}                                    & \multicolumn{2}{c|}{Receiver}                                           \\ 
    \hline 
            & Male                           & Female                        & Male                             & Female                              \\ 
    \hline 
    Like    & 41.1 \% (750) & 58.9\% (1079) & 40.9\% (743) & 59.1\% (1074)  \\
    Comment & 40.9\%(750)    & 59.1\% (1087) & 40.7\% (735)   & 59.3\% (1073)      \\
    Tags   & 40.6\% (704)  & 59.4\% (1031) & 37.3\% (540) & 62.7\% (910)    \\
   Users   & 40.8\% (763)  & 59.2\% (1104) & 40.9\% (748) & 59.1\% (1082)    \\
    \hline
    \end{tabular}
    \end{footnotesize}
    \label{tab:basic_fb}
\end{table}

\textbf{Limitations.}
The data is collected over a long time frame but for a rather small and homogeneous group of students. Thus it is not straight-forward to generalize from the findings in this dataset to the general population of Facebook users.

\subsection{Effect Analysis }
\noindent\textbf{Single-hop Analysis.}
Figures~\ref{fig:fb_ccdftag_intensity} and \ref{fig:fb_ccdftag_indegree} present CCDFs for tags. Table~\ref{T:u_test_fb} shows the statistical confidence for different user groups. Different from comments and likes, females dominate males in terms of receiving intensity and in-degree. Females show stronger visibility in any rank than males, almost twice of males. From the sending side, females are also more active in endorsing others. 
As users on Facebook are closer than on Instagram, users with high visibility are more active in endorsement. However, we observe the glass ceiling for tag's outdegree, indicating top-ranked males endorse more actively than their female counterparts.\footnote{Please refer to \cite{online} for the complete results of the Facebook dataset.}

 \begin{shaded*}
     \noindent\textit{\textbf{Takeaway}: In terms of received like, comment and tag intensities, females are more visible than males in any rank, showing no sign of glass ceiling.}
 \end{shaded*}
\begin{figure*}[t]
\hfill
\begin{subfigure}{.20\linewidth}    
    \centering
    \includegraphics[width=3.5cm,height=2.6cm]{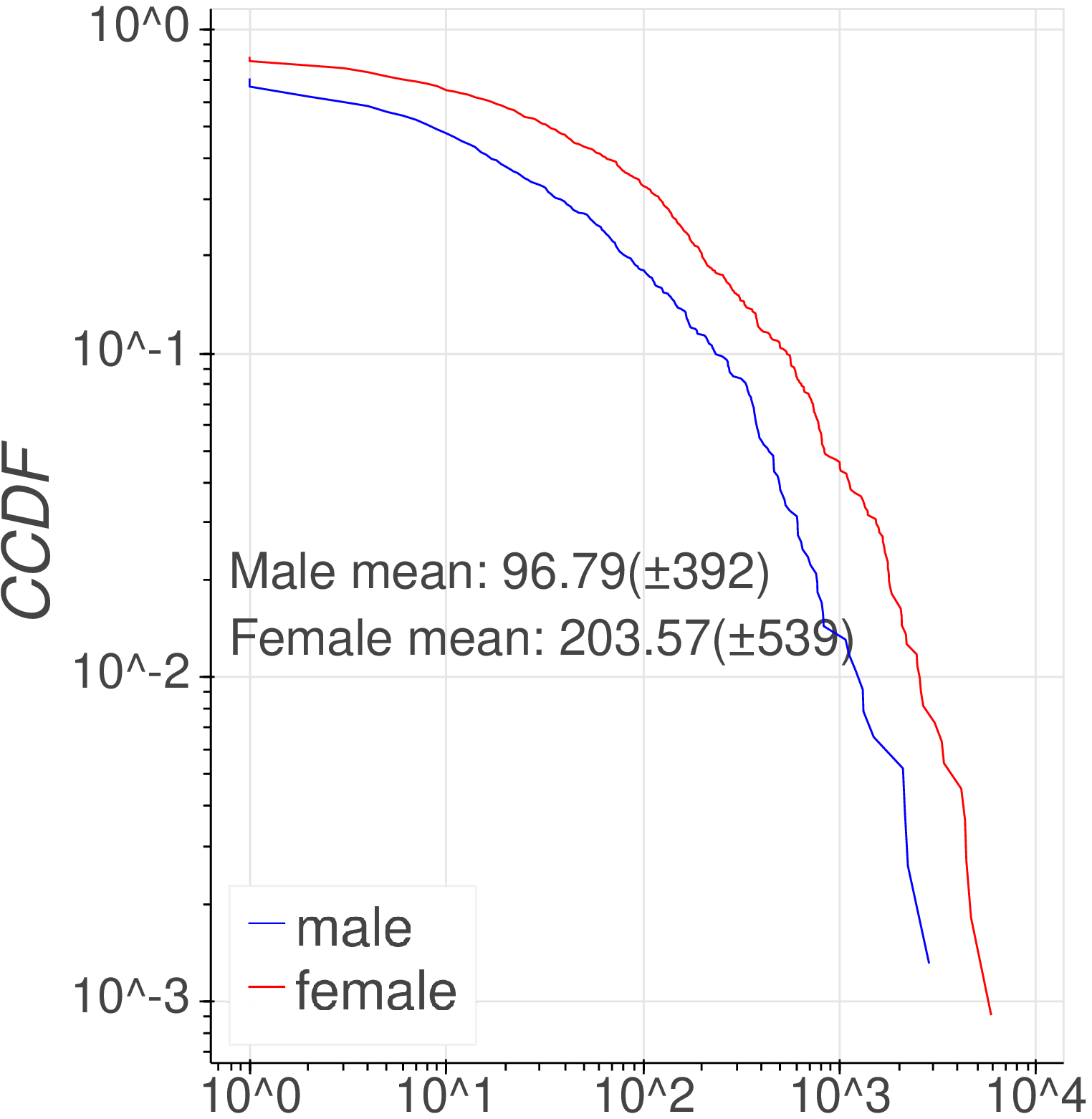}
    \caption{Received (Tags)}
    \label{fig:fb_ccdftag_intensity}
\end{subfigure}
    \hfill
\begin{subfigure}{.20\linewidth}    
    \centering
    \includegraphics[width=3.5cm,height=2.6cm]{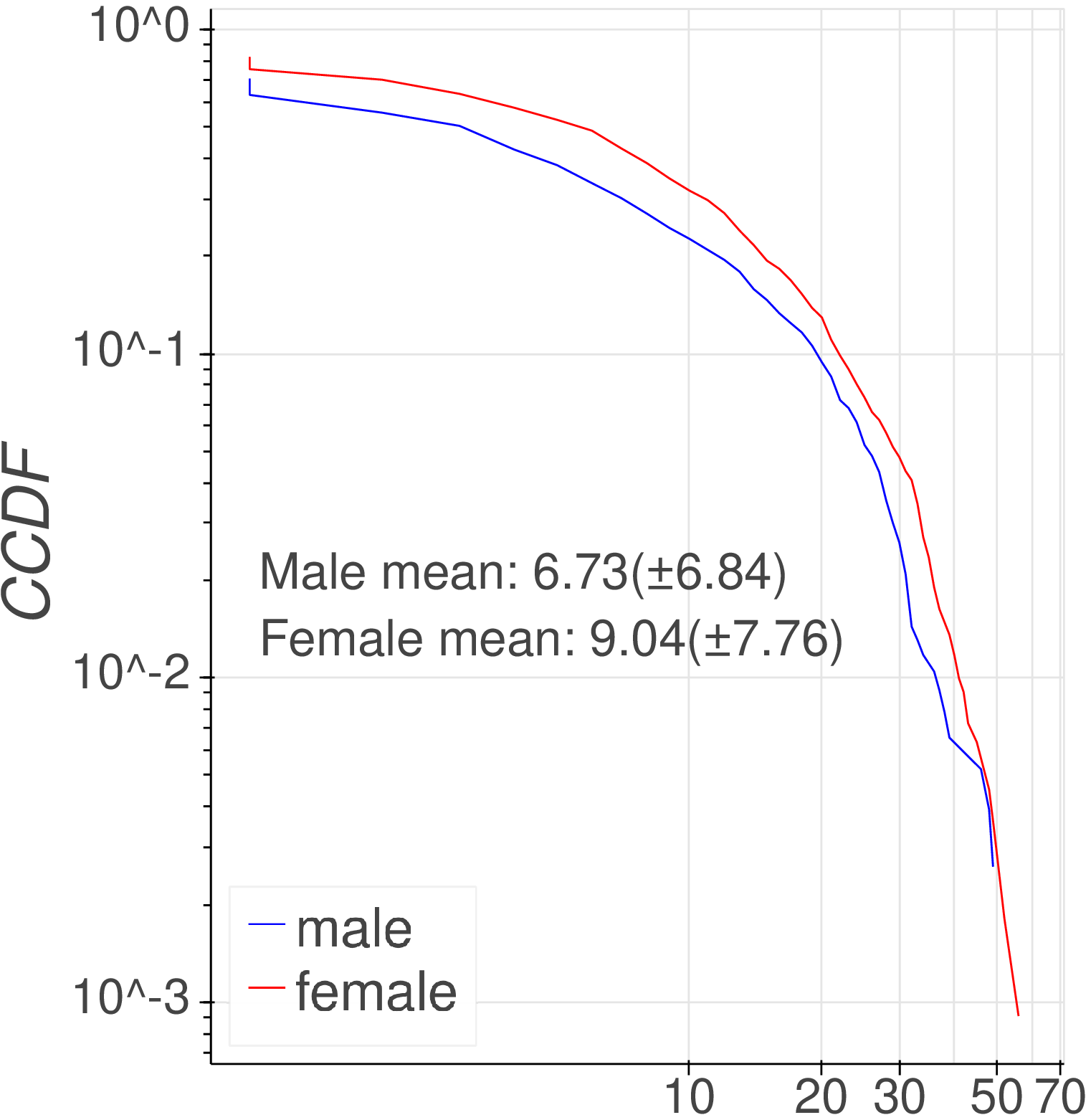}
    \caption{In-degree (Tags)}
    \label{fig:fb_ccdftag_indegree}
\end{subfigure}
\hfill
\begin{subfigure}{.20\linewidth}
    \centering
    \includegraphics[width=3.5cm,height=2.6cm]{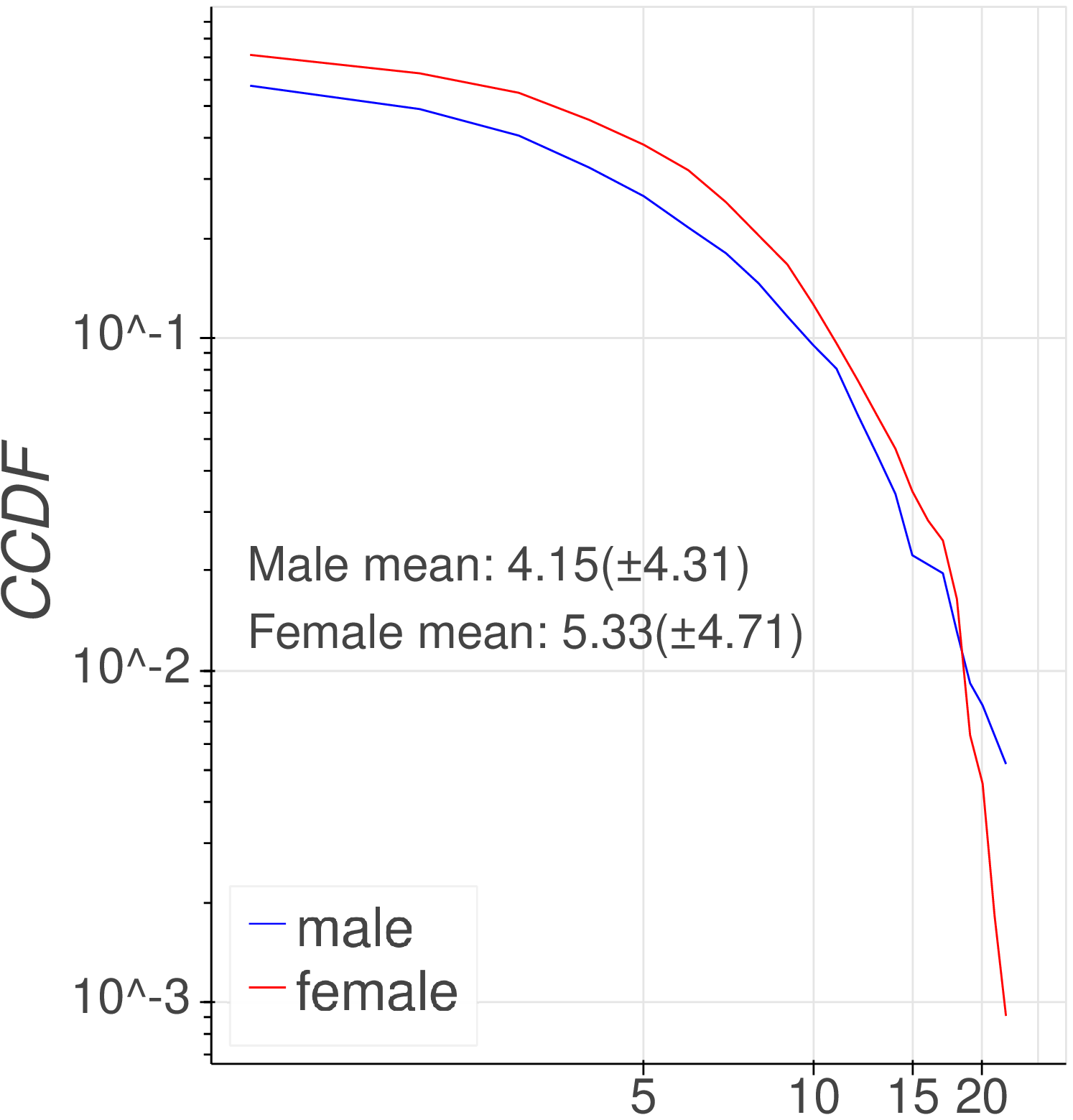}
    \caption{HI-index (Tags)}
    \label{fig:fb_ccdftag_HI-index}
\end{subfigure}
   \hfill
\begin{subfigure}{.20\linewidth}    
    \centering
    \includegraphics[width=3.5cm,height=2.8cm]{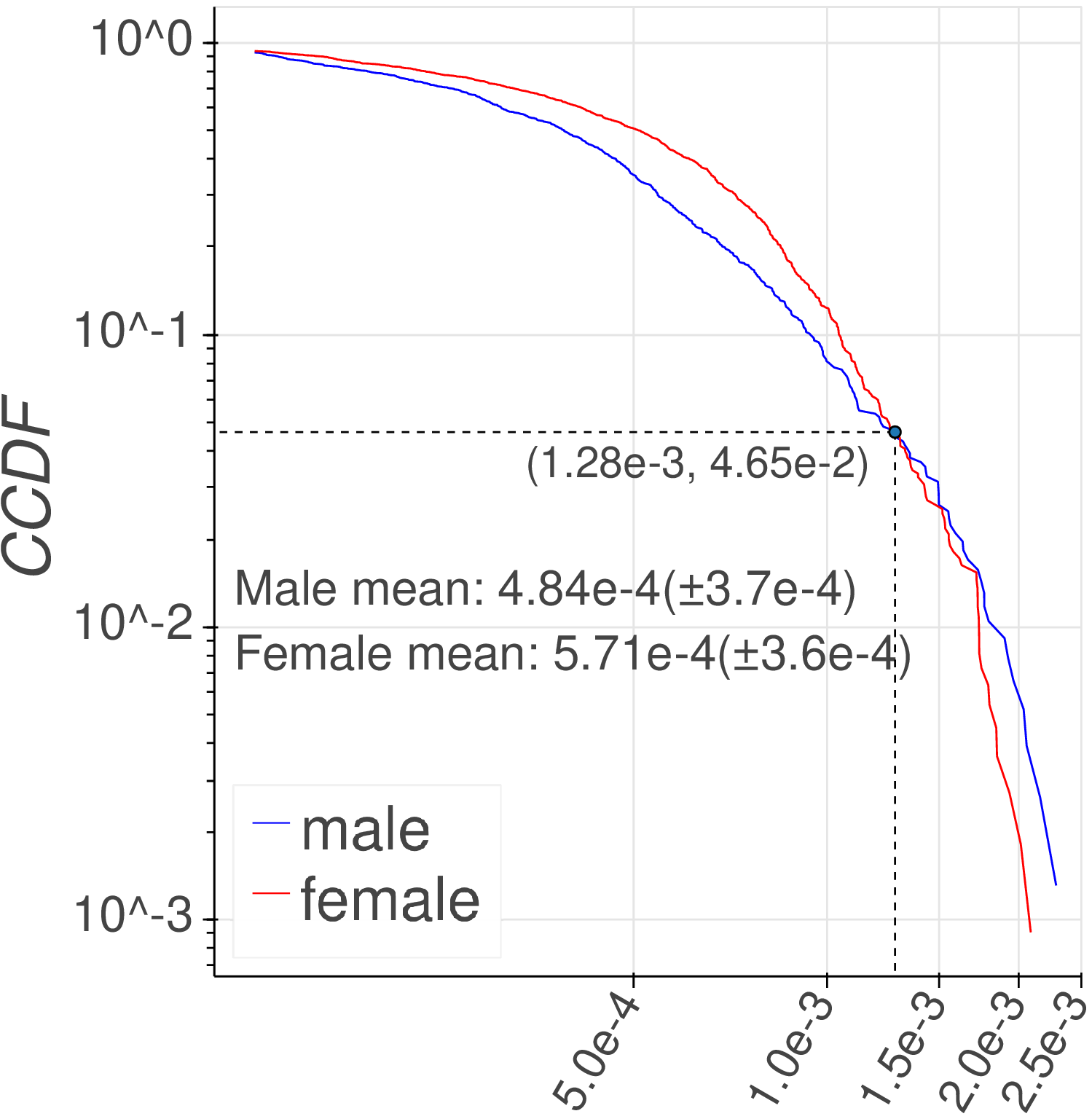}
    \vspace{-2.5mm}
    \caption{PageRank (Tags)}
    \label{fig:fb_ccdftag_pagerank}
\end{subfigure}
    \hfill

\vspace{-.3cm}
\caption{Visibility (received) in terms of Tag on Facebook.}
\label{fig:fb_ccdf}
\vspace{-.2cm}
\end{figure*}

\noindent\textbf{Multi-hop Analysis.}
We focus on the multi-hop analysis for tags due to its contrasting trend from the single-hop analysis. While females clearly dominate males in terms of intensity and indegree,  strong glass-ceiling effects are observed in PageRank, shown in Fig.~\ref{fig:fb_ccdftag_pagerank}. Zooming into Fig.~\ref{fig:fb_ccdftag_HI-index}, the mean value of the HI-index is $4.15$ and $5.33$ for males and females, respectively, indicating that on average male/female users interact with slightly more than 4 and 5 other users respectively who also interact with $4$ or $5$ others. However the top $1\%$ of male users have an HI-index of up to $23.62$ users, while top female users only reach $23.89$.

 \begin{shaded*}
     \noindent\textit{\textbf{Takeaway}: There is a strong glass ceiling for tags, from multi-hop measures but not sing-hop on Facebook dataset. }
 \end{shaded*}

Finally, we also summarize in Table~\ref{table:glassceiling} if a glass ceiling effect is exhibited by the different centrality measures and types of interactions for both the Instagram and the Facebook datasets.

\begin{table}[t]
\small
\setlength{\tabcolsep}{0.6em}
    \caption{Mann-Whitney U test: the p-value on Facebook.}
    \label{T:u_test_fb}
    \vspace{-.3cm}
    \centering
    \begin{tabular}{|cc|ccc|}
    \hline & & Top 10\% & Top 1\% & Top 0.1\% \\\hline
    \multirow{4}{*}{Tag} & Rec. intensity & 1.812e-16*** & 7.626e-4*** & 0.0*** \\
    & In degree & 5.423e-05*** & 0.094 & 0.0*** \\
    & HI-index & 5.423e-05*** & 0.094 & 0.0*** \\
    & PageRank & 0.099 & 0.007** & 0.0*** \\ \hline
    \multicolumn{5}{l}{*, **, and *** denote $p\leq 0.05$, $p \leq 0.01$, and $p\leq0.001$, respectively.} \\
    \end{tabular}
\end{table}

\begin{table}
\vspace{.1cm}
\caption{Summary of the glass ceiling on Instagram and Facebook. \xmark~ indicates that a glass ceiling effect is observed.}
\label{table:glassceiling}
\vspace{-.3cm}
\centering
\resizebox{\columnwidth}{!}{
\begin{tabular}{c|c|c;{1pt/1pt}c|c;{1pt/1pt}c;{1pt/1pt}c|}
\cline{3-7}
\multicolumn{2}{c|}{\multirow{2}{*}{}}                                     & \multicolumn{2}{c|}{{\cellcolor[rgb]{0.69,0.69,0.69}}\textbf{Instagram}}                                   & \multicolumn{3}{c|}{{\cellcolor[rgb]{0.69,0.69,0.69}}\textbf{Facebook}}                                                                               \\
\cline{3-7}
\multicolumn{2}{c|}{}                                                      & {\cellcolor{gray!10}}\textit{~~~Likes~~} & {\cellcolor{gray!10}}\textit{Comments} & {\cellcolor{gray!10}}\textit{~~~Likes~~} & {\cellcolor{gray!10}}\textit{Comments} & {\cellcolor{gray!10}}\textit{~~~Tags~~~} \\
\hline
\multirow{4}{*}{\rotatebox[origin=c]{90}{\begin{tabular}[c]{@{}c@{}}~\textbf{Single}\\\textbf{hop}\end{tabular}}} & 
                                                   \begin{tabular}[c]{@{}c@{}}Rec. intensity\end{tabular}  & \xmark & & & & \\
\cdashline{2-7}[1pt/1pt]
                                                & \begin{tabular}[c]{@{}c@{}}Sen. intensity\end{tabular}   & & & & & \\
\cdashline{2-7}[1pt/1pt]
                                                & \begin{tabular}[c]{@{}c@{}}In degree\end{tabular}        & \xmark & \xmark & & \xmark & \\
\cdashline{2-7}[1pt/1pt]
                                                & \begin{tabular}[c]{@{}c@{}}Out degree\end{tabular}       &  \xmark & \xmark & & & \xmark \\
\hline
\multirow{2}{*}{\rotatebox[origin=c]{90}{\begin{tabular}[c]{@{}c@{}}~\textbf{Multi}\\\textbf{hop}\end{tabular}}}  & 
                                                   \begin{tabular}[c]{@{}c@{}}HI-index\end{tabular}        & & & \xmark & \xmark & \\
\cdashline{2-7}[1pt/1pt]
                                                & \begin{tabular}[c]{@{}c@{}}PageRank\end{tabular}         & \xmark & & & & \xmark\\
\hline
\end{tabular}
}
\vspace{.2cm}

\end{table}

\section{\seeding}\label{sec:seeding}
In this section, we address the research question of how to optimize the information spread via a seeding strategy to the targeted demographic group, i.e., how to meet a target gender ratio.\footnote{We use the target gender ratio and the target female ratio interchangeably here.} The solution can be applied to commercial and governmental campaigns that aim to optimally reach a certain percentage of females~\cite{True:ISQ:2002:networkdiffusion} or other groups in society. Concretely, a marketing campaign may aim to select a fixed number of influencers to try out a product before its formal release such that the corresponding product information will be maximally spread at a ratio of 30\% females. This is especially challenging when the targeted gender ratio deviates from the original ratio in the population. Formally, the influence maximization with the target gender ratio constraint is formulated as follows.

\begin{definition}[Disparity Influence Maximization (DIM)]
Given a social network $G=(V,E)$, a diffusion model, a seed group size $K$, a target gender ratio $\zeta$, and an error margin $e$, the problem is to select a seed group $S \subseteq V$ with $\left| S \right|=K$ to maximize the influence spread under the constraint that the female ratio in the influenced users is $\zeta$ within an error margin $e$.
\end{definition}

\begin{thm}\label{thm:np}
It is NP-hard to approximate the Disparity Influence Maximization problem (DIM) within any factor unless $P=NP$.
\end{thm}

\begin{proof}
Please refer to ~\cite{online} for the proof.
\end{proof}

\subsection{\seeding Framework}
We propose the \textit{Disparity Seeding} framework, which selects influential females and males according to a given seed group size and target gender ratio that can be far off from the population ratio. Our disparity framework is composed of two phases, illustrated in Figure~\ref{fig:flowchart}: ranking users and deciding seeding ratio. First, we identify influential users by two gender-aware metrics, i.e., \emph{Target HI-index} and \emph{Embedding index}, to select the most influential node achieving the target gender ratio. Then, we estimate at what proportion to allocate the available seeds to males and females, abbreviated as the seeding ratio, based on their ranks. Our earlier analysis indicates there is a clear gap between the two. For instance, there is a higher percentage of very influential males even though there is a lower percentage of males in the population.\footnote{Note that our Disparity Seeding framework is a general method which can also support other sensitive attributes, e.g., race or age.}

\begin{figure}[t]
\centering
\begin{subfigure}{1\linewidth}
    \centering
    \includegraphics[width=\columnwidth]{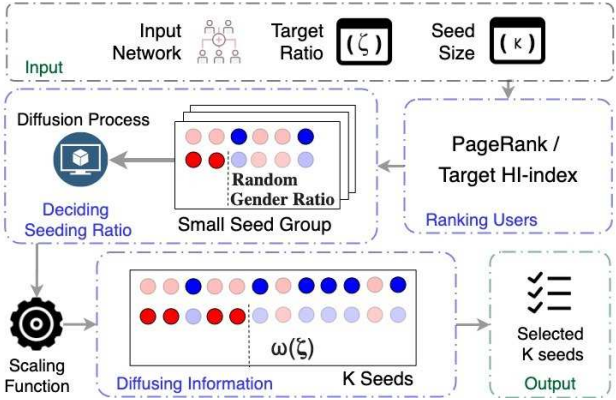}
\end{subfigure}
\vspace{-.3cm}
    \caption{The framework illustration of Disparity Seeding.}
    \label{fig:flowchart}

\end{figure}

\subsubsection{Ranking Users}  \label{sec:target_hi_index}

To rank users, we introduce two gender-aware ranking measures, \emph{Target HI-index} and \emph{Embedding index}.

First, Target HI-index, based on HI-index, favors the users having direct neighbors with a gender ratio close to the given target. In other words, the users who interact with direct neighbors with a dissimilar female ratio to the target ratio $\zeta$ are penalized. 
\begin{definition}
\emph{Target HI-index} of a user $v_i$ is defined as $v_i$'s HI-index but penalized by the difference between $\zeta$ and the female ratio of $v_i$'s neighbors from which the HI-index is derived. 
Let $N^\text{F}(v_i,n)$ denote the number of $v_i$'s direct female neighbors interacting with others at least $n$ times. The Target HI-index of $v_i$ is formulated as
\begin{align}\small
TH(v_i,\zeta)=H(v_i)
\cdot \Big(1-  |\frac{N^\text{F}\big(v_i,H(v_i)\big)}{N\big(v_i,H(v_i)\big)} - \zeta |\Big),
\end{align}
is the penalty for the female ratio of $v_i$'s direct neighbors having at least $H(v_i)$ interactions with others not satisfying $\zeta$. A larger difference between $\zeta$ and the female ratio in $N\big(v_i,H(v_i)\big)$ results in a greater penalty on $H(v_i)$.
\end{definition}

On the other hand, as Graph Neural Networks (GNNs) have been celebrating successes in various domains by modeling the long-term dependencies between nodes in a graph, we also introduce the \emph{Embedding index} with a novel GNN model, namely \emph{Influence Graph Neural Network (InfGNN)} model. In contrast to previous ranking methods, which rely on oversimplified assumptions, remain highly sensitive to their hyperparameters~\cite{aral2012identifying}, and thus can provide inappropriate estimations compared to actual cascades~\cite{pei2018theories}, InfGNN can estimate the influence dynamically and accommodate rich node attributes to support diverse goal, e.g., target gender ratio.\footnote{Note that the Embedding index is only adopted on Facebook since this dataset contains user profiles which are required for the initialization of GNN \cite{kipf2016semi}.} In the following, we first present the definition of the \emph{Embedding index} and then describe the design of the InfGNN model. 
\begin{definition}
  \emph{Embedding index} of a user is learned from the social network by preserving the proximity between users by InfGNN, which also penalizes the difference between $\zeta$ and the female ratio. 
\end{definition}
To effectively estimate the Embedding index, InfGNN consists of two primary components: 1) stacked GNN layers to derive the node embeddings and 2) an influence predictor to estimate the influence of the learned embeddings. Given the $l^{th}$ hidden feature $\mathbf{H}^l$ of nodes in the graph, GNNs update the  $(l+1)^{th}$ latent features of node $v_i \in V$, denoted by $\mathbf{h}^{l+1}_i\in\mathbf{H}^{l+1}$, by aggregating the features of $v_i$'s neighboring nodes $N(v_i)$, which can be written as follows. 
\begin{equation}\small\label{eq:agg}
    \mathbf{h}^{l+1}_{i}= ReLU (Attn^l (\mathbf{h}_{i}^{l},\mathbf{h}_j^{l},\forall v_{j} \in N(v_i))),
\end{equation}
where $Attn^l(\cdot)$ is the attention mechanism~\cite{velivckovic2017graph} for aggregation, and $ReLU(\cdot)$ is a non-linear activation function~\cite{kipf2016semi}.  After stacking the  $L^{th}$ layer, InfGNNs adopts the final hidden layer $\mathbf{H}^L$ as the embedding $\mathbf{h}_i$ of each node $v_i$. Then, we adopt skip-gram objective~\cite{mikolov2013distributed} as our proximity loss to structural information on graph $G$, i.e., 
\begin{equation}\small
   L_{prox} = -\sum_{v_i \in V} (\sum_{v_j \in N(v_i)} \sigma(\mathbf{h}^{T}_i  \mathbf{h}_j) + \mathbb{E}_{v_j \sim P_{N}(\mathbf{v}_i)} [\sigma (-\mathbf{h}^{T}_i  \mathbf{h}j)]),
\end{equation}
where $\sigma$ is the sigmoid function and $P_{N}(\cdot)$ is the distribution for negative sampling of users. After deriving the node embedding, the next step is to predict the influence score $s_i \in \mathbf{s}$ of each node $v_i$ by
\begin{equation}\small
    s_i = \sigma (\mathbf{c}^T\mathbf{h}_{i}),
\end{equation} 
where $\mathbf{c}$ is a trainable vector to calculate the importance of each node. In contrast to traditional GNNs, which usually require a great effort in collecting labels (e.g., the importance of each node)~\cite{kipf2016semi,velivckovic2017graph}, we introduce influence loss $L_{inf}$  to train InfGNN in a self-supervised manner,
\begin{equation}\small\label{eq:linf}
   L_{inf} =\sum_{v_i \in V} (\norm{s_i-s'_i}_2 - \mathbb{E}_{v_j \sim P_{N}(\mathbf{v}_i)} (\norm{s_i-s_j}_2)).
\end{equation}
Note that the first part of the Eq. (\ref{eq:linf}) minimizes the error between the self and estimated scores, and the second part is used to distinguished the influence between each node by negative sampling. The estimated score $s'_i$ is derived from $v_i$'s neighborhoods, i.e.,
\begin{equation}\small
    s'_i = \sum_{v_j \in N(v_i)}\frac{\exp(\mathbf{a}^T[\mathbf{h}_i,\mathbf{h}_j]) s_j}{\sum_{v_k \in N(v_j)}\exp(\mathbf{a}^T[\mathbf{h}_j,\mathbf{h}_k])}.
\end{equation}
where $\mathbf{a}$ is a trainable vector to measure the influence between nodes. Unlike conventional GNN~\cite{velivckovic2017graph} only aggregates the $1$-hop neighborhood information and thereby suffer from oversmoothing~\cite{kipf2016semi}, i.e., every node has the same  embedding and thereby lead to a similar influence score, our InfGNN estimate the influence by considering the $2$-hop information. While previous statistical measurements~\cite{Alonso:JI:2009:Hindex} require to identify the influence by some fixed hyperparameters to exploit a specific property of the graph, InfGNN calculates the importance of each node dynamically and thereby is more general. Our overall objective becomes, 
\begin{equation}\small
   L = L_{prox} + \lambda_1  L_{inf}  + \lambda_2 \norm{\mathbf{s}}_1 + \lambda_3 \norm{\mathbf{h}}_2,
\end{equation} 
Note that $\norm{\mathbf{s}}_1$ is the $l_1$ regularization, which discretizes the output distribution to force the model to concentrate the influence on a few nodes and $\norm{\mathbf{h}}_2$ is the $l_2$ regularization of node embeddings. $\lambda_1$, $\lambda_2$ and $\lambda_3$ are the hyperparameters to determine the trade-off between the proximity and the influence score. We adopt SGD~\cite{bottou2012stochastic} to minimize the above loss function and obtain the embedding and influx score of each node. Finally, our Embedding index (EI) of the node $v_i$ is defined as follows.
\begin{equation}\small
   EI(v_i) = s_i \cdot D_{in}(v_i) \cdot \Big(1-|\frac{N^\text{F}\big(v_i,n)\big)}{N\big(v_i,n)\big)} - \zeta |\Big),
\end{equation} 
where $D_{in}(v_i)$ is the indegree of node $v_i$. Intuitively, the influence score $s_i$ represent the spreading intention of a user, while the indegree $D_{in}(v_i)$ is the number of connections. Note that  $\Big(1-|\frac{N^\text{F}\big(v_i,n)\big)}{N\big(v_i,n)\big)} - \zeta |\Big)$ penalizes to $v_i$ if it violates the target gender ratio.\footnote{The threshold interacting times $n$ of Embedding index is a hyperparameter which is set to $3$, empirically.}

\begin{figure}[t]
\centering
\begin{subfigure}{.49\linewidth}
    \centering
    \includegraphics[width=.9\columnwidth]{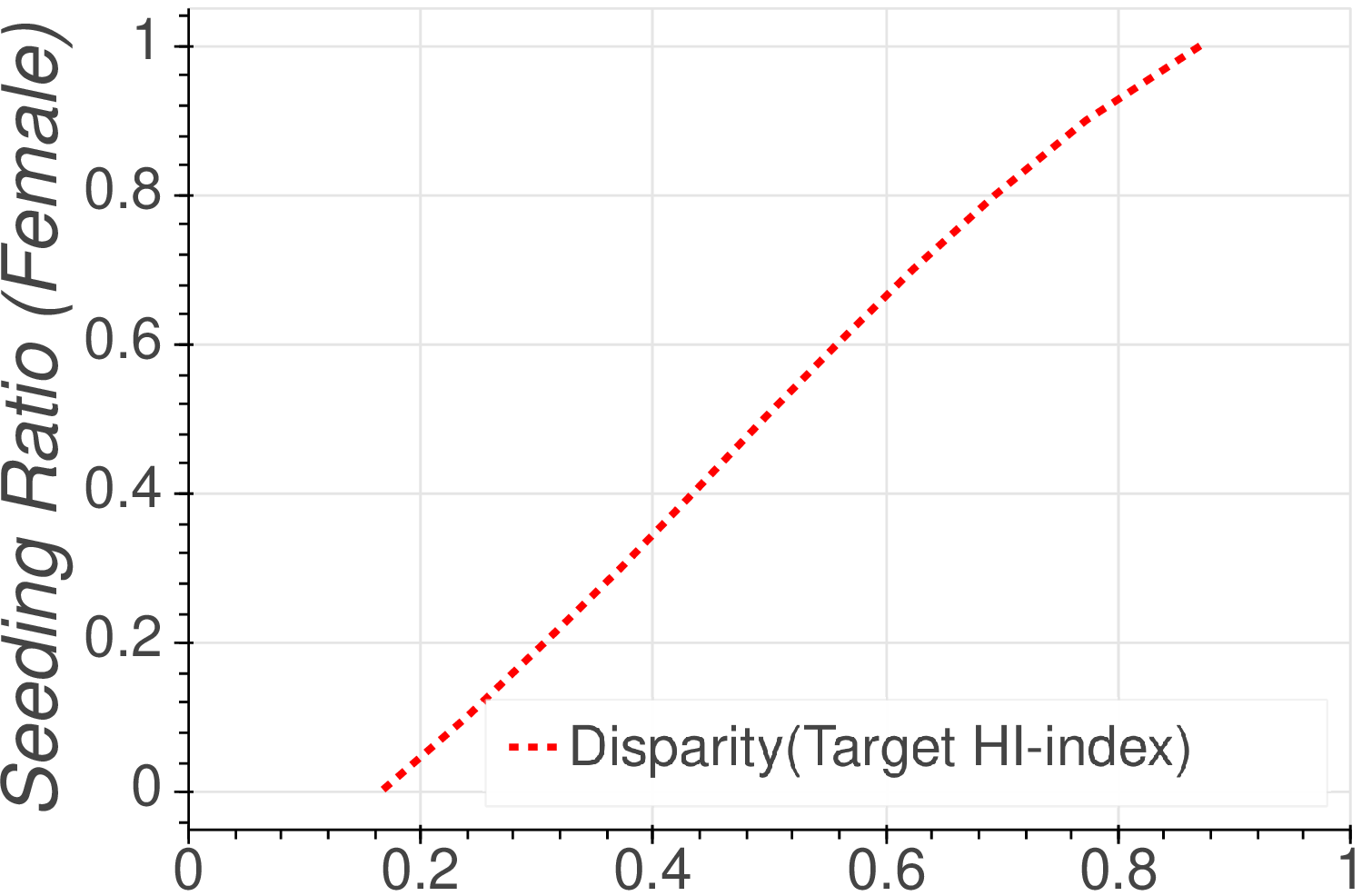}
    \caption{Target ratio (Likes)}
    \label{fig:seeding2target_like}
\end{subfigure}
    \hfill
\begin{subfigure}{.49\linewidth}
    \centering
    \includegraphics[width=.9\columnwidth]{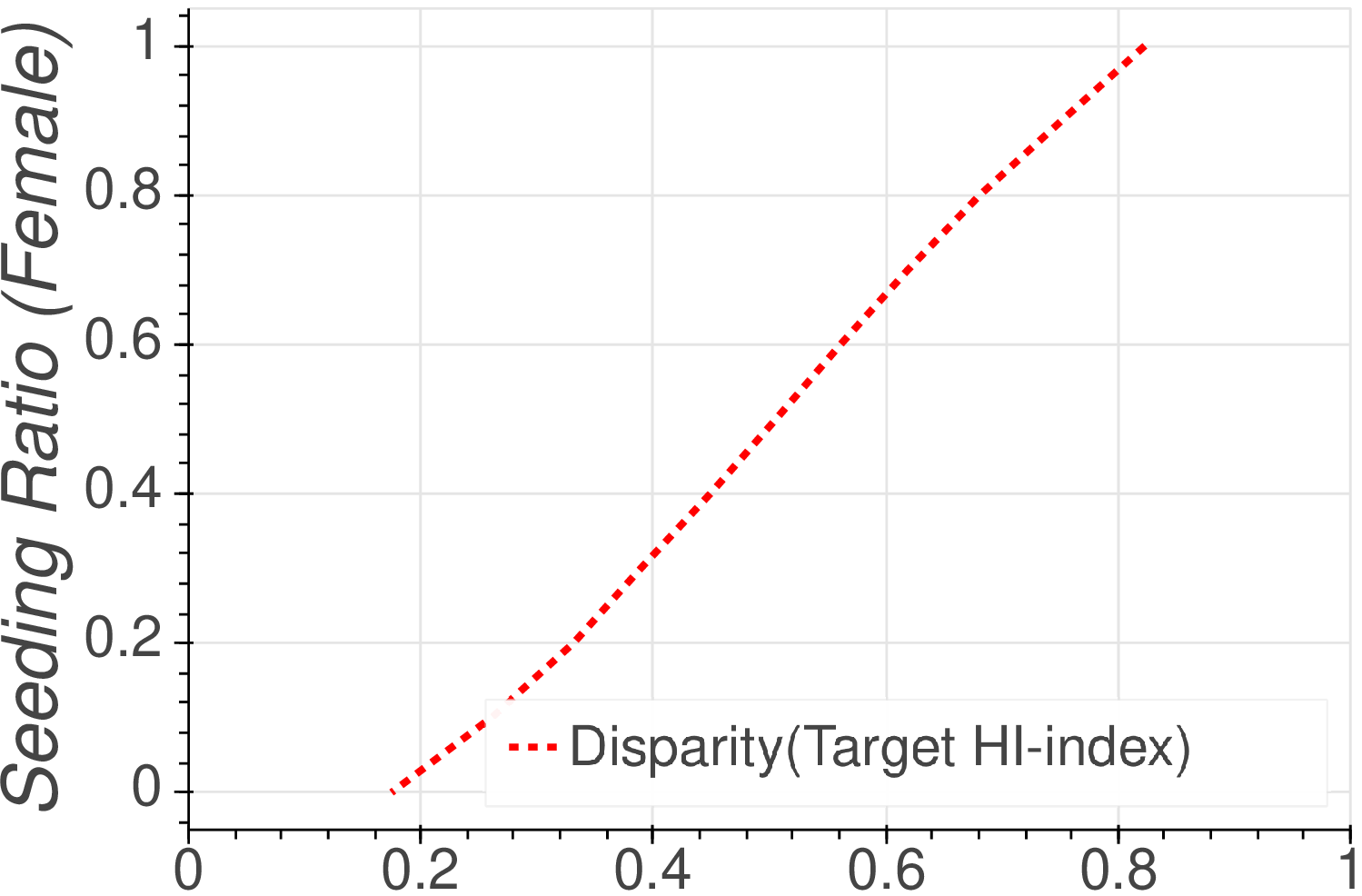}
    \caption{Target ratio (Comments)}
    \label{fig:seeding2target_comment}
\end{subfigure}
\vspace{-.3cm}
\caption{The learned scaling function $\omega$.}
\label{fig:seeding2target}

\end{figure}

\begin{figure*}[t]
    \centering
    \includegraphics[scale=0.5]{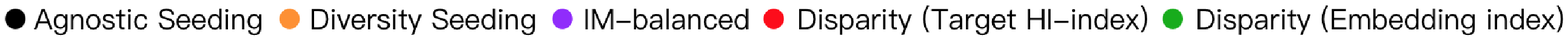}
\vspace{-3mm}
\end{figure*}
\begin{figure*}[t]
 \centering
 \begin{subfigure}{.19\linewidth}
     \centering
     \includegraphics[width=\columnwidth]{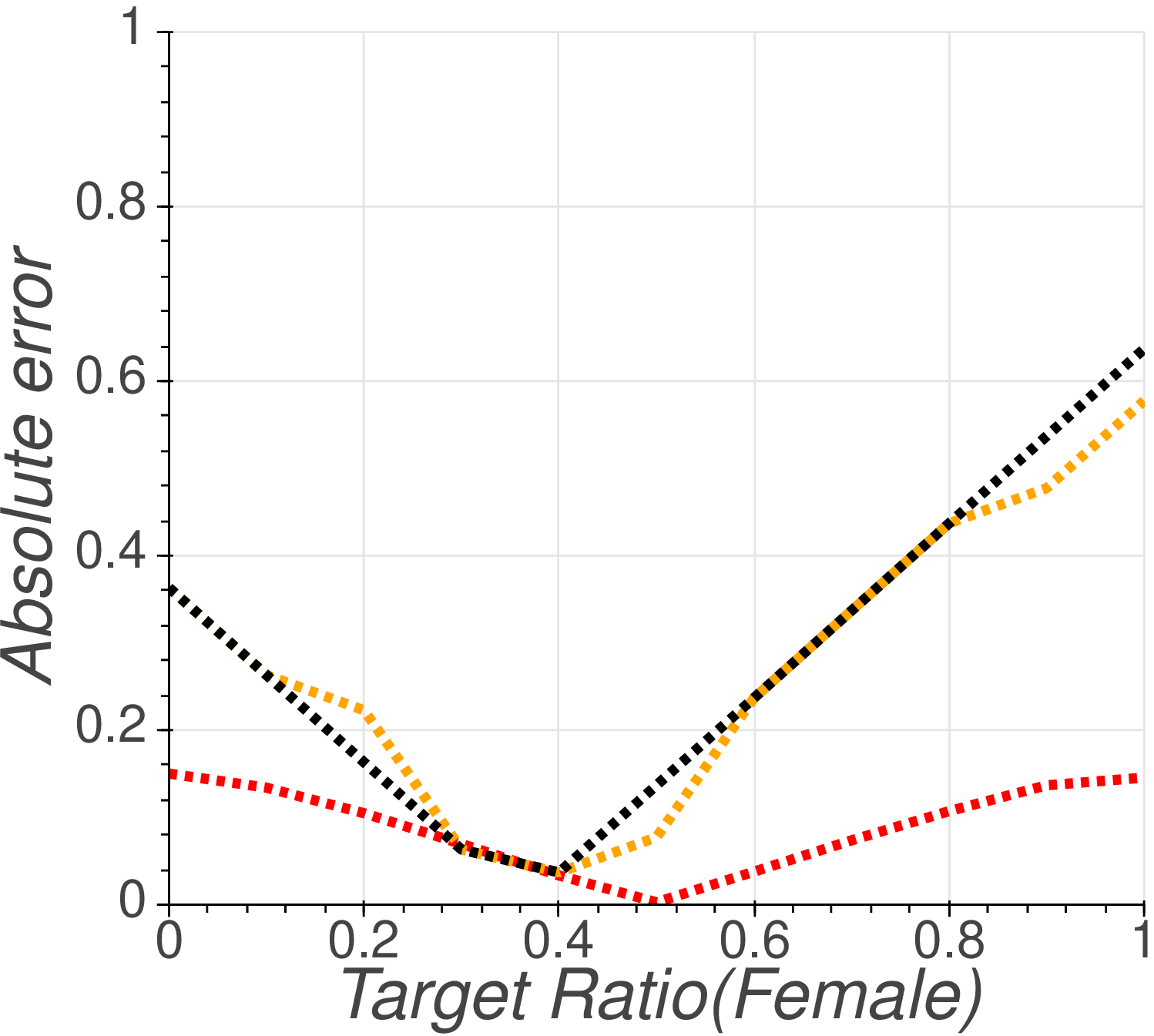}
     \caption{Instagram (Comments)}
     \label{fig:parentofront_comment}
      \end{subfigure}
 \hfill
 \begin{subfigure}{.19\linewidth}
     \centering
     \includegraphics[width=\columnwidth]{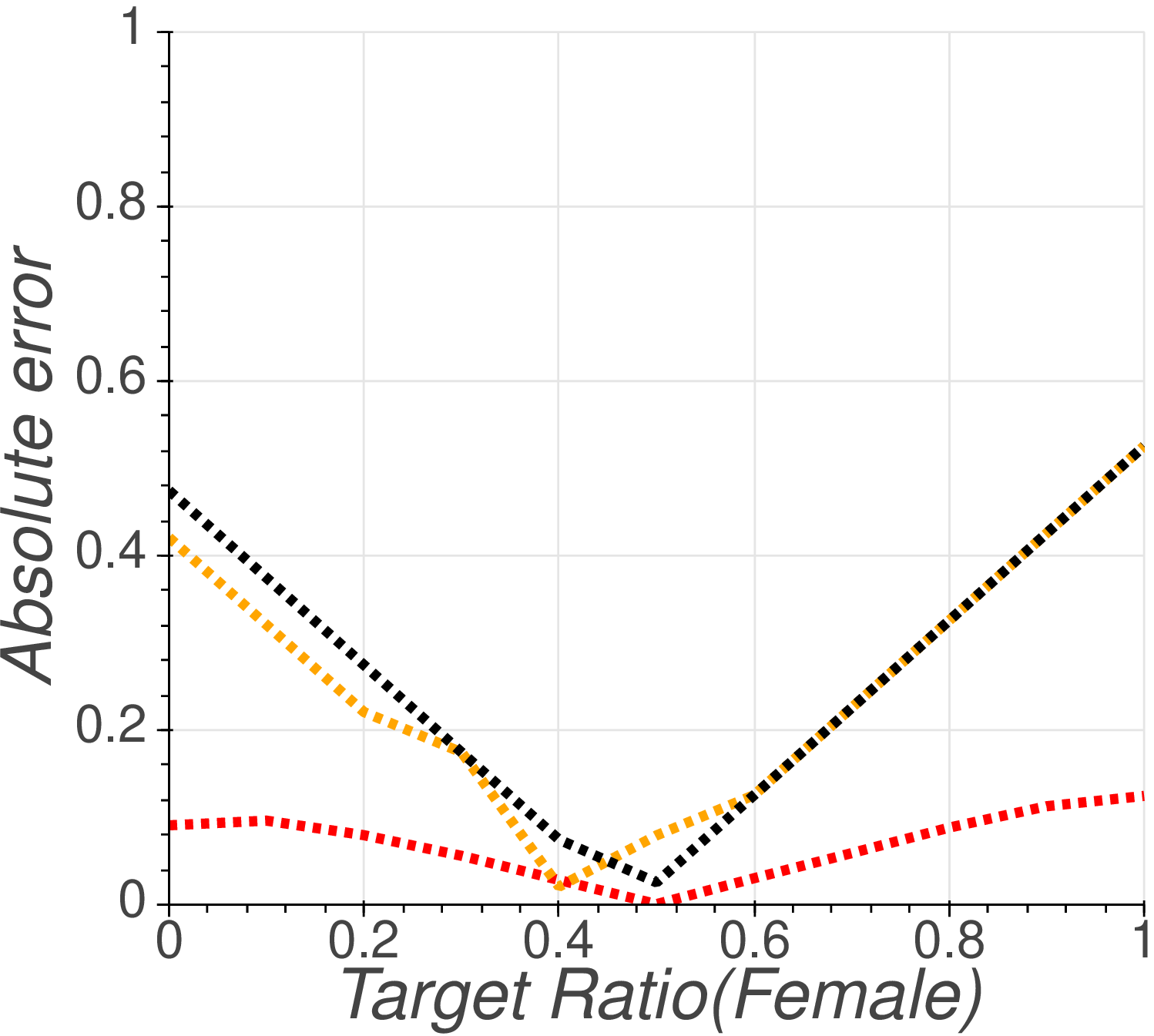}
     \caption{Instagram (Likes)}
     \label{fig:parentofront_like}
 \end{subfigure}
  \begin{subfigure}{.19\linewidth}
     \centering
     \includegraphics[width=\columnwidth]{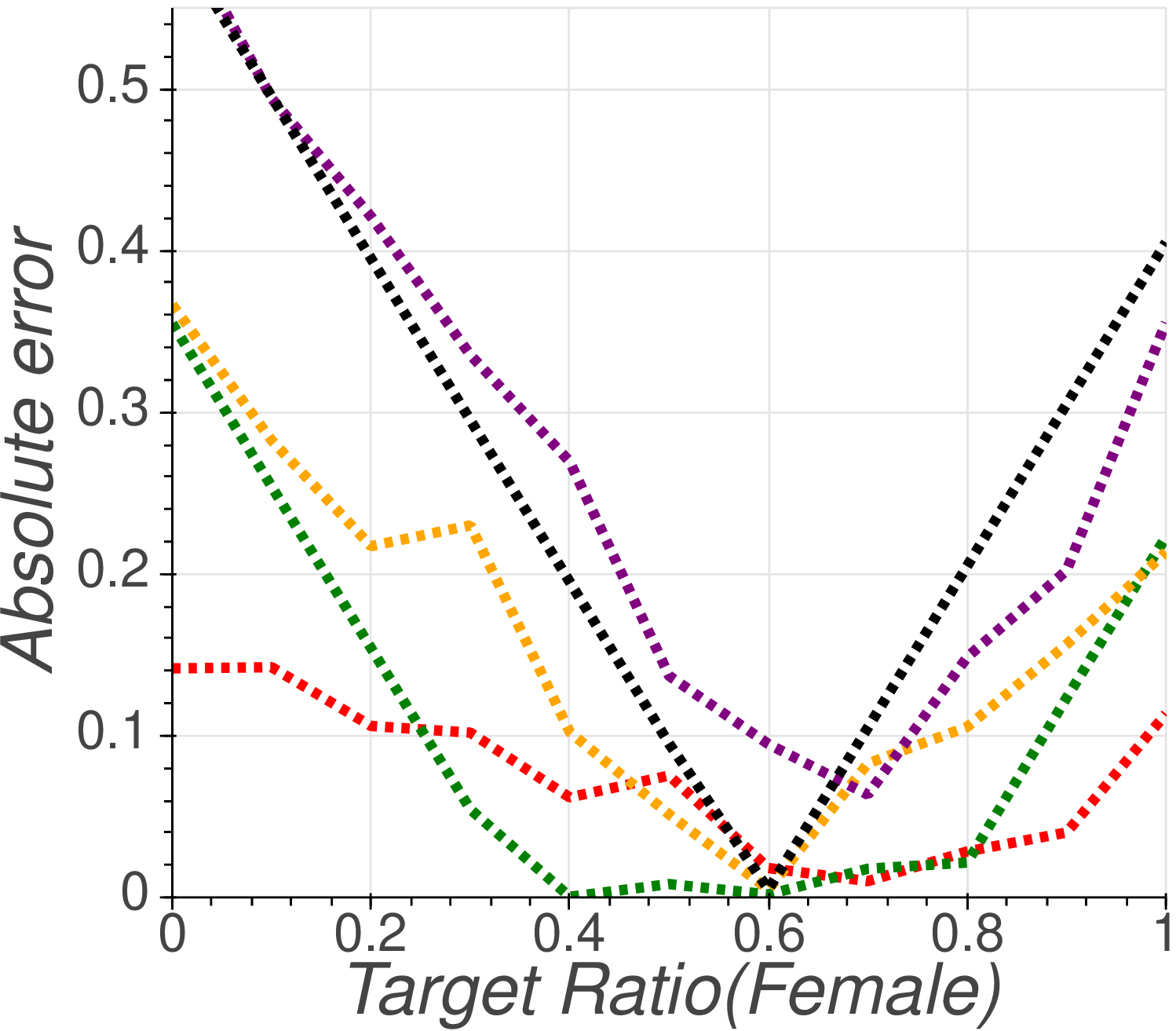}
     \caption{Facebook (Comments)}
     \label{fig:fb_parentofront_comment}
      \end{subfigure}
 \hfill 
 \begin{subfigure}{.19\linewidth}
     \centering
     \includegraphics[width=\columnwidth]{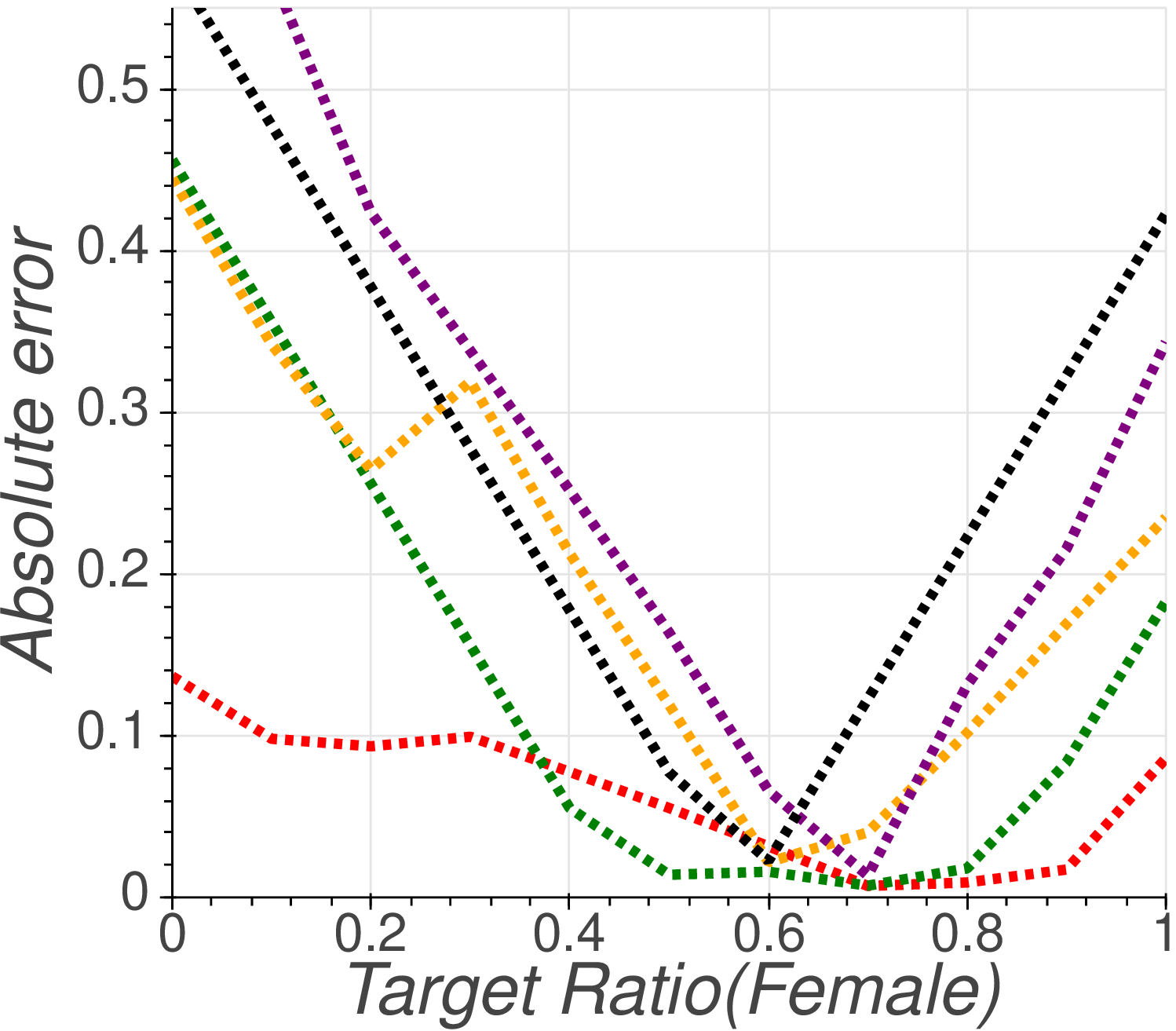}
     \caption{Facebook (Likes) }
     \label{fig:fb_parentofront_like}
 \end{subfigure}
 \hfill
 \begin{subfigure}{.19\linewidth}
     \centering
     \includegraphics[width=\columnwidth]{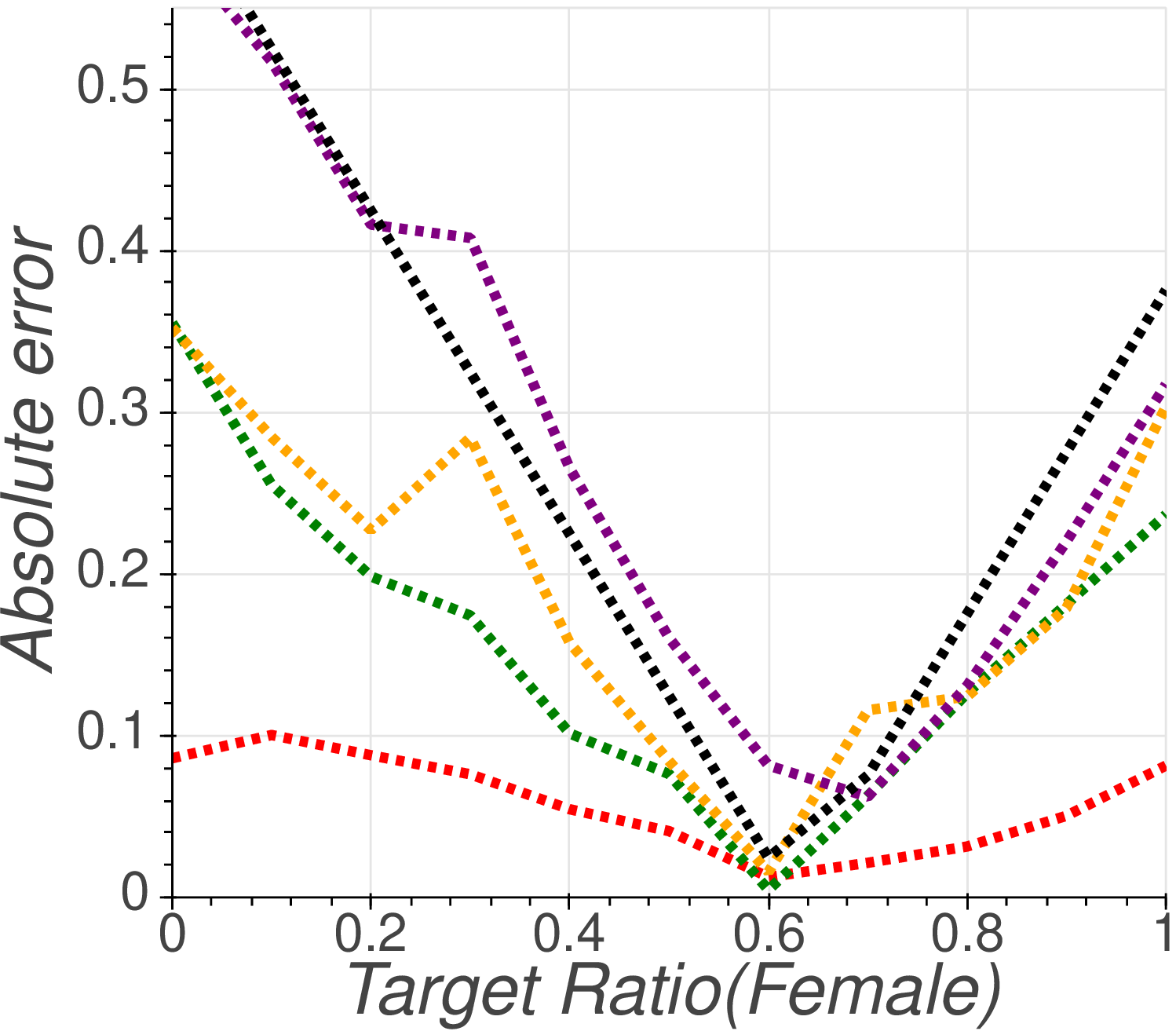}
     \caption{Facebook (Tags)}
     \label{fig:fb_parentofront_tag}
 \end{subfigure}
 
 \hfill
 \vspace{-.3cm}
 \caption{Evaluation on Instagram and Facebook.}
 \label{fig:roi_r1u_FB}
 \vspace{-.2cm}
 \end{figure*}

\subsubsection{Deciding Seeding Ratio} \label{sec:omega}
Here, we search for a scaling function $\omega: \zeta \rightarrow r$ that can map the target ratio $\zeta$ into the actual seeding ratio $r \in [0, 1]$ based on the ranking results. We use a simulated diffusion process\footnote{To obtain the dependency properly, the information propagation is simulated according to the specified diffusion process in the problem.} to capture the dependency between the target and seeding ratio and thus learn $\omega$. 
Specifically, a number of seed groups with different gender ratios are selected for information diffusion. 
Larger seed groups can better capture the dependency between the target and seeding ratio, whereas smaller groups shorten the simulation and learning time. The optimal choice of the size of such seed groups exceeds the current scope of the paper. The scaling function $\omega$ is determined by the seeding ratio and the female ratio of users adopting the information in each simulated result.\footnote{Essentially $\omega$ needs to be empirically learned for each ranking algorithm separately.} We explain how to derive $\omega$ from the simulation data as follows.

Figures~\ref{fig:seeding2target_like} and~\ref{fig:seeding2target_comment} illustrate the $\omega$ function of Target HI-index for liking and commenting networks on Instagram, respectively. Each point represents how to choose a seeding ratio (y-axis) given a target ratio (x-axis). The point is iteratively determined through simulations. The $\omega$ function of Target HI-index centrality has a wider range from $0.16$ to $0.86$ of target ratios (in the x-axis), showing higher flexibility in accommodating different disparity ratios. This observation holds for both likes and comments.

Finally, we select a total of $K$ seeds based on the learned $\omega$ function and the target ratio. They are essentially the top $\omega(\zeta)\cdot K$ females and the top $(1-\omega(\zeta))\cdot K$ males in their gender group. We conduct a final run of simulation using such seeds for both centrality measures for ranking, namely Target HI-index and Embedding index, and compare the information spread and resulting in each gender ratio. We return the seeds selected by the centrality measure to maximize the overall information spread and achieve the target gender ratio within a certain error margin.

\subsection{Evaluation}

\subsubsection{Setup}

The demo system of \seeding is online.\footnote{The demo system of \seeding: \url{http://bit.ly/DisparitySeeding}.}
We compare two variants of \seeding (i.e., Target HI-Index and Embedding index) with the state-of-the-art approaches: diversity seeding ~\cite{Stoica:WWW:2020:SeedingDiver} and IM-balanced~\cite{gershtein2018balanced}.\footnote{The codes and datasets are available in \cite{online}.} Note that both diversity seeding and IM-balanced are not designed for the target gender ratio $\zeta$. To be aware of the target gender ratio, we implement diversity seeding and IM-balanced as follows. Diversity seeding selects a value between $\zeta$ and the gender ratio of the top $K$ users with the highest in-degree as the seeding ratio such that the influence is maximized. IM-balanced greedily selects seeds to maximize the influence while ensuring that the influence on females is at least $\zeta$ of the optimal influence on females. 
Following~\cite{Stoica:WWW:2020:SeedingDiver}, the diffusion is simulated through the Independent Cascade diffusion model~\cite{Kempe2003}, where the probability for user~$v_i$ to influence user~$v_j$ is set as the number of likes/comments/tags $v_i$ gives to $v_j$ over the total number of likes/comments/tags that received by user~$v_i$~\cite{Nguyen:SIGMOD:2016:stop}. 

We evaluate the performance with varying target gender ratios. Specifically, the performance metrics include 1) the absolute error, i.e., $\left| s -\zeta \right|$, where $s$ is the spreading gender ratio and 2) the influence spread.
The seed group size $K$ is $5000$ and $100$ for Instagram and Facebook, respectively; the error margin is $20\%$ for both datasets.
For \seeding, the size of sampled seeds is $1,000$ and $20$ for Instagram and Facebook, respectively. Each simulation result is averaged over $10,000$ samples. All experiments are run on an HP DL580 server with an Intel 2.40GHz CPU and 4GB RAM.

\subsubsection{Simulation Results on Instagram.}
Figures~\ref{fig:parentofront_comment} and~\ref{fig:parentofront_like} manifest the absolute errors (y-axis) under \seeding with Target HI-index and diversity seeding with a varying target gender ratio (x-axis).\footnote{\seeding with Embedding index is not adopted on Instagram since the Instagram dataset lacks the user profiles which are required for the initialization of GNN. IM-balanced is omitted since it employs the greedy approach and cannot report the results within $12$ hours due to the large scale of Instagram.} Note that an approach with a smaller absolute error is more capable of fulfilling the requirement of the target gender ratio. We have two observations. First, for all target gender ratios, \seeding has the smallest absolute errors because \seeding carefully decides the seeding ratio by capturing the dependency between the target and the seeding ratios. Besides, target HI-index factors the target ratio into ranking and penalizes users who do not fulfill the requirement, adding more flexibility in accommodating extreme target ratios. By contrast, even though diversity seeding takes the target ratio as its input, it only looks for a ratio to maximize the influence spread rather than minimizing the error between the spreading and the target ratios. 
Second, the absolute errors are the smallest for both approaches when the target ratio is between $0.5$ and $0.7$ because the female ratio on Instagram is $52.15\%$, making both methods easier to satisfy the target gender ratio around $50\%$.

\subsubsection{Simulation Results on Facebook.}
Figures~\ref{fig:fb_parentofront_comment}-\ref{fig:fb_parentofront_tag} show the absolute errors on Facebook under different target ratios. In addition to Target HI-index, since Facebook contains user profiles, we further evaluate Embedding index, which ranks users according to both network structure and user profiles. The proposed \seeding with either of the ranking measures can achieve the smallest absolute errors, whereas diversity seeding and IM-balanced are unable to achieve. Note that IM-balanced can only satisfy the target ratio between $0.5$ and $0.8$ since the female ratio on Facebook is $59.09\%$, showing that the guaranteed ratio of the optimal influence on females is not related to the target gender ratio. Comparing Target HI-index and the Embedding index, one can see that Target HI-index is better than the Embedding index, i.e., Target HI-index is located within the error margin of $20\%$ for all target gender ratios.

Table~\ref{T:inf} compares the influence spread on Instagram and Facebook when the target ratio $\zeta$ is satisfied within the error margin of 20\%.\footnote{\seeding with Embedding index is ``n/a'' on Instagram since this dataset lacks the user profiles for the initialization of GNN. IM-balanced is not included since it cannot report the results on Instagram within a reasonable time and its spread on Facebook does not satisfy the target ratio within the error margin of $20\%$.} For Facebook, \seeding with Embedding index outperforms the others since it carefully exploits user profiles on the Facebook dataset to identify influential users. For Instagram, \seeding with Target HI-index has the influence spread comparable to diversity seeding. Nevertheless, Target Hi-index has much smaller absolute errors than diversity seeding for all $\zeta$.

\begin{table}[t]
\small
\setlength{\tabcolsep}{0.4em}
    \caption{The influence spread on Instagram and Facebook.}
    \label{T:inf}
    \vspace{-.3cm}
    \centering
    \begin{tabular}{|c|ccc|cc|}
    \hline
    \multirow{2}{*}{} & \multicolumn{3}{c|}{Facebook} & \multicolumn{2}{c|}{Instagram} \\
     & Comment & Like & Tag & Comment & Like\\ \hline
    Target ratio $\zeta$  & 0.9 & 0.9 & 0.9 & 0.5 & 0.5\\ \hline
    Target HI-index & 222.2 & 202.5 & 241.2 & \textbf{8542.6} & 7508.7  \\
    Embedding index & \textbf{232.8} & \textbf{255.3} & \textbf{263.8} & n/a & n/a  \\
    Diversity seeding & 187.8 & 224.9 & 243.8 & 8516.14 & \textbf{7531.73}  \\
    Agnostic seeding & 175.5 & 182.1 & 229.9 & 8469.28 & 7523.47  \\
    \hline
    \end{tabular}
\end{table}

\section{Conclusion}
Leveraging Instagram and Facebook interaction data, we investigated usage patterns and the manifestation of the glass ceiling in different interaction types on social media. We discovered correlations between gender and both high visibility and popular endorsement by jointly considering degrees and interaction intensity of direct and indirect interactions with traditional and novel measures. Motivated by the gender disparity exhibited in online social networks, we proposed a \seeding framework that aims to maximize the information spread and reach a gender target ratio, which may differ from the original ratio in the population. The core step of \seeding applies the proposed centrality measures to rank influential users, namely Target HI-index and Embedding index, and selects a suitable seed set through a simulation-based learning approach. Our evaluation results show that \seeding can achieve the target gender ratio in an agile manner and increase the information spread better than state-of-the-art parity seeding algorithms. The general design of \seeding can be extended to include additional information. It can also be applied to select seed sets to take under-privileged minorities into account and enhance their influence. 
\section*{Acknowledgement}
This work has been partly funded by the Swiss National
Science Foundation NRP75 project 407540\_167266, by MOST through grants 110-2221-E-001-014-MY3 and 109-2221-E-001-017-MY2, and by the Institute for Information Industry under Contract 110-0922.
We thank to National Center for High-performance Computing (NCHC) of National Applied Research Laboratories (NARLabs) in Taiwan for providing computational and storage resources.

\clearpage
\bibliographystyle{ACM-Reference-Format}
\balance
\bibliography{sample-base}

\end{document}